\documentclass[aps,twocolumn,pra,floatfix,groupedaddress,nofootinbib,notitlepage,showpacs,superscriptaddress]{revtex4-1}

\usepackage{amsthm}
\usepackage{latexsym}
\usepackage{float}
\usepackage{appendix}
\usepackage{epstopdf}
\usepackage{physics}
\usepackage[colorlinks=true,linkcolor=blue,citecolor=green,plainpages=false,pdfpagelabels]{hyperref}
\usepackage[normalem]{ulem}
\usepackage{newlfont}
\usepackage{amsfonts}

\usepackage{epsfig}
\usepackage{mathrsfs}
\usepackage[thinc]{esdiff}

\newtheorem{lemma}{Lemma}
\newtheorem{proposition}{Proposition}

\newtheorem{theorem}{Theorem}

\DeclareOldFontCommand{\rm}{\normalfont\rmfamily}{\mathrm}
\usepackage{bigints}

\usepackage{braket}

\usepackage{caption}
\usepackage{subcaption}
\usepackage{setspace}
\usepackage{amsmath}
\usepackage{times,bm,bbm,bbold,graphicx,graphics,amssymb,amsmath,amsfonts,dsfont,hyperref,color}
\usepackage{mathrsfs,cancel}
\usepackage[utf8]{inputenc}
\usepackage{soul,xcolor}
\setstcolor{red}
\definecolor{mycolor}{rgb}{0.1, 0.1, 0.7}
\usepackage{dsfont}
\usepackage{url}
\usepackage{soul}
\usepackage[vlined,ruled]{algorithm2e}

\DeclareMathAlphabet{\mathpzc}{OT1}{pzc}%
{m}{it}
\hypersetup{
	plainpages=true,
	breaklinks=true,
	hypertexnames=false,
	pageanchor=true,
	colorlinks=true,
	linkcolor={mycolor},
	citecolor={mycolor},
	urlcolor={mycolor},
	anchorcolor={black}
}

\begin{document}

\title{Uncertainty Relations in Pre- and Post-Selected Systems}

\author{Sahil}
\email{sahilmd@imsc.res.in}
\affiliation{Optics and Quantum Information Group, The Institute of Mathematical Sciences,
CIT Campus, Taramani, Chennai 600113, India}
\affiliation{Homi Bhabha National Institute, Training School Complex, Anushakti Nagar, Mumbai 400085, India}

\author{Sohail}
\email{sohail@hri.res.in}
\affiliation{Quantum Information \& Computation Group, Harish-Chandra Research Institute, A CI of Homi Bhabha national Institute, Chhatnag Road, Jhunsi, Prayagraj - 211019, India.}

\author{Sibasish Ghosh}
\email{sibasish@imsc.res.in}
\affiliation{Optics and Quantum Information Group, The Institute of Mathematical Sciences,
CIT Campus, Taramani, Chennai 600113, India}
\affiliation{Homi Bhabha National Institute, Training School Complex, Anushakti Nagar, Mumbai 400085, India}

\begin{abstract}
 In this work, we  derive Robertson-Heisenberg like uncertainty relation   for two incompatible observables in a pre- and post-selected (PPS)  system. The newly defined standard deviation and the uncertainty relation in the PPS system have physical meanings which we  present here.  We demonstrate two unusual properties   in the PPS system using  our uncertainty relation.  First,  for  commuting observables, the lower bound  of the uncertainty relation in the PPS system does not become zero  even if the initially prepared state i.e., pre-selection   is the eigenstate of  both the observables when  specific post-selections are considered. This implies that for such case, two commuting observables can disturb each other's measurement results which is in fully contrast with  the Robertson-Heisenberg  uncertainty relation.  Secondly, unlike the standard quantum system, the PPS system makes it feasible to  prepare sharply a quantum state  (pre-selection)  for non-commuting observables {(to be detailed in the main text)}.   Some  applications of uncertainty and uncertainty relation in the PPS system are provided: $(i)$ detection of mixedness of an unknown  state, $(ii)$ stronger uncertainty relation in the standard quantum system, ($iii$) ``purely quantum uncertainty relation"  that is,  the uncertainty relation which is not affected (i.e., neither increasing nor decreasing) under the classical mixing  of quantum states, $(iv)$ state dependent tighter uncertainty relation in the standard quantum system, and  $(v)$ tighter upper bound for the out-of-time-order correlation function.
\end{abstract}

\maketitle
\section{introduction}
The uncertainty relation, which Heisenberg discovered, is one of the most well-known scientific findings \cite{Heisenberg-1927,zurek-1983}. It asserts that it is impossible to accurately measure the position and the momentum of a particle. In other words, measuring the position of a particle  always affects   the momentum of that particle and vice versa. Robertson developed the uncertainty relation known as  ``\emph{Robertson-Heisenberg Uncertainty Relation}" (RHUR)  \cite{Robertson-Heisenberg-1929} in the very later years to describe the difficulty of jointly sharp preparation of a quantum state {(see Ref.  \cite{explanation-1} for the notion of {\it sharp preparation})} for incompatible observables. This relation not only limits the joint sharp preparation for   non-commuting observables  but also proved it's usefulness:  to formulate  quantum mechanics \cite{Busch-2007,Lahti-1987}, for entanglement detection \cite{Hofmann-2003,Guhne-2004}, for the security analysis of quantum key distribution in quantum cryptography \cite{Fuchs-1996}, as a fundamental building block for quantum mechanics and quantum gravity \cite{mjwhall-2005}, etc.\par
On the one side, we have the standard quantum systems where the RHUR hold while pre- and post-selected (PPS) systems, on the other side, are  different kind of quantum mechanical systems that were developed by Aharonov, Bergmann, and Lebowitz (ABL) \cite{ABL1964, Aharonov-Vaidman-2008, Kofman-Nori} to address the issue of temporal asymmetry in quantum mechanics.    Recently, in the references \cite{Gammelmark-2013,Tan-2015}, the authors   generalized    the probabilities of obtaining the measurement results of an observable in a PPS system given by ABL \cite{ABL1964}.\par 
  In the later years,   Aharonov,  Albert, and Vaidman (AAV)  \cite{AAV-1988} introduced the notion of ``weak value"  defined as 
\begin{align}
\braket{A_w}^{\phi}_{\psi}=\frac{\braket{\phi|A|\psi}}{\braket{\phi|\psi}},\label{intro-1}
\end{align}
 in a pre- and post-selected system when the observable $A$ is  measured weakly.  Here, $\ket{\psi}$ and $\ket{\phi}$ are pre- and post-selected states, respectively.   Weak values have  strange features for being complex and its real part can  lie outside the max-min range of the eigenvalues of the operator of interest when the pre- and post-selections are nearly orthogonal.\par
 {In order to obtain the real and imaginary parts of the weak value of $A$ \cite{Lundeen-Resch-2005, Jozsa-2007}, first  the system of interest and a pointer (ancilla) is prepared  in the product state   $\ket{\psi}\otimes \ket{\xi}$. Then  the system-pointer is evolved  under  the global unitary  $U=exp(-iHt)$, where  $H=gA\otimes P_x$ is the von Nuemann Hamiltonian, $A$ is the measurement operator of the system, $P_x$ is the pointer's momentum observable,  `$g$' is the coupling coefficient between system and  pointer and `$t$' is the interaction time. Now after the time evolution of the  system-pointer, the system is projected to $\ket{\phi}$ and  as a result, the state of the   pointer collapses to the unnormalized state $\ket{\widetilde{\xi}_{\phi}} \approx {{\braket{\phi|\psi}}}(1-igt\braket{A_w}^{\phi}_{\psi}P_x)\ket{\xi}$ in the limit $g\ll 1$, i.e., weak interaction. Now, it can be shown that  the  average position and momentum shifts of the pointer in  state  $\ket{{\xi}_{\phi}}=\frac{\ket{\widetilde{\xi}_{\phi}}}{\sqrt{\braket{\widetilde{\xi}_{\phi}|\widetilde{\xi}_{\phi}}}}$  are  
$\braket{X}_{{\xi}_{\phi}}=gt \,\mathcal{R}{e}(\braket{A_w}^{\phi}_{\psi})$ and  $\braket{P_x}_{{\xi}_{\phi}}=\frac{gt}{2\sigma^2}\,\,\mathcal{I}{m}(\braket{A_w}^{\phi}_{\psi})$, respectively 
with the Gaussian pointer  $\braket{x|\xi}=\left(\frac{1}{\sqrt{2\pi}\sigma}\right)^{1/2}e^{-x^2/4\sigma^2}$, $\sigma$ is the rms width of the position distribution  $|\braket{x|\xi}|^2$ of the pointer and thus providing the full knowledge of the weak value of $A$.}\par
Recently, a lot of attention was paid to these aspects \cite{Omar-2014, Hosten-2008, Lundeen-2011, Lundeen-Bamber-2012, Thekkadath-2016, Cheshire-2013, Denkmayr, Das-Pati-2020, Sahil-2021, Pan-2019, Turek-2020, Aharonov-2001, Lundeen-Steinberg-2009, Yokota-Imoto,  Quantum_Paradoxes, Yakir-2012, Solli-2004,Rozema-2012, Pati-2015, Nirala-2019, Pati-2020, Goswami-2019}.  The measurements involving weak values are known as “Weak Measurements” or “weak PPS measurements”.   Since it depends on probabilistic post-selection $\ket{\phi}$, a weak value can be thought of as conditional expectation value. Moreover,  when the post-selection is same as pre-selection i.e., $\ket{\phi}=\ket{\psi}$, it becomes
  \begin{align}
\braket{A}_{\psi}=\braket{\psi|A|\psi},\label{intro-2}
\end{align}
the expectation value   in the standard quantum system. The PPS systems can therefore be thought of as being more general than the so-called standard quantum systems. \par
As pre- and post-selected systems are already useful practically as well as fundamentally, then  an immediate question can be asked whether  there exists   any uncertainty relation like {the RHUR} which can give the limitations  on joint sharp preparation of  the given pre- and post-selected states when non-commuting observables are measured.\par
 In the present study, we demonstrate the existence of such uncertainty relations in PPS systems, which are    expected as PPS systems are more generalised versions of standard ones. We first define the standard deviation of an observable in the PPS system for the given pre- and post-selections with geometrical as well as physical interpretations. After that,  we derive our main result of this paper  ``\emph{uncertainty relations in  pre- and post-selected systems}" using the well known Cauchy-Schwarz inequality. \par
  We provide  the following physical applications  of our results: ($i$) detection of the purity of an unknown state  of   any quantum systems (e.g., qubit, qutrit, two qubit, qutrit-qubit, etc)  using two different definitions of the uncertainty of an observable in the PPS system, ($ii$) stronger uncertainty relation in the standard quantum system (i.e., the uncertainty relation that can not be made trivial or the lower bound can not be made zero for almost all possible choices of initially prepared systems) using the uncertainty relation in the PPS system,  ($iii$) purely quantum uncertainty relations that is,  the uncertainty relations which are not affected (i.e., neither increasing nor decreasing) under the classical mixing  of quantum states using the uncertainty relations  in PPS systems.  $(iv)$ state dependent tighter uncertainty relation in the standard system  by introducing the idea of  post-selection, and  finally $(v)$ tighter   upper bound for the out-of-time-order correlation function. Moreover, as the RHUR has a plenty of applications, uncertainty relation in the PPS systems can also be  applied  in quantum optics, information, technologies, etc.\par
{This paper is organized as follows. In Sec. \ref{II}, we discuss uncertainty relations in standard quantum systems. In Sec. \ref{III}, we derive our main results of this paper.  Application of our results  are given  in Sec. \ref{IV}  and finally we  conclude our work in Sec. \ref{V}.}

\section{Uncertainty relation in standard quantum system }\label{II}
 In this section, we first  interpret the standard deviation of an observable in  standard quantum systems  from  geometrical as well as  information-theoretic perspective. For establishing the standard deviation in a PPS system, we will  introduce a similar interpretation. The RHUR's well-known interpretation is also provided here. \par
 \subsection{Standard deviation}
  We consider the system Hilbert space to be $\mathcal{H}$ and let $\ket{\psi}$ be a state vector in $\mathcal{H}$.  Due to the probabilistic nature of the measurement outcomes of the observable  $A$,  the uncertainty in the measurement is defined as the standard deviation:
 \begin{align}
  \braket{\Delta{A}}_{\psi}=\sqrt{\braket{\psi|A^2|\psi} - \braket{\psi|A|\psi}^2}\label{RHUR-1}.
\end{align}
\emph{\textbf{Geometric interpretation.---}}
Standard deviation can be  given  a geometrical interpretation using the following {proposition}. 
{\begin{proposition}
     If $\ket{\psi}\in \mathcal{H}$ is an initially prepared state of a standard quantum system and $A\in \mathcal{L}(\mathcal{H})$  is a hermitian operator, then we can decompose $A\ket{\psi}\in\mathcal{H}$ as 
     \begin{align}
{A}\ket{\psi}=\braket{{A}}_{\psi}\ket{\psi} + \braket{\Delta A}_{\psi} \ket{\psi^{\perp}_{A}},\label{RHUR-2}
\end{align}
where $\ket{\psi^{\perp}_{A}}=\frac{1}{\braket{\Delta A}_{\psi}}({A}-\braket{A}_{\psi})\ket{\psi}$,  and $\braket{A}_{\psi}=\braket{\psi |{A}|\psi}$. Eq. (\ref{RHUR-2}) is sometimes known as the ``Aharonov–Vaidman identity" \cite{Lev-formula}.
\end{proposition}
\begin{proof}
 Let $A\ket{\psi}$ and $\ket{\psi}$  are two non-orthogonal state vectors. Using Gram-Schmidt orthogonalization process, we find the unnormalized  state vector $\ket{\widetilde{\psi}^{\perp}_{A}}\in\mathcal{H}$ orthogonal to $\ket{\psi}$ as 
\begin{equation}
 \ket{\widetilde{\psi}_{A}^{\perp}}= {A}\ket{\psi} - \frac{(\braket{\psi |{A})|\psi} }{\braket{\psi|\psi}}\ket{\psi}=({A}-\braket{A}_{\psi})\ket{\psi},\label{Lev-1}
\end{equation}
where $\braket{\psi|\psi}=1$ and after normalization, Eq. (\ref{Lev-1}) becomes
\begin{align}
{A}\ket{\psi}=\braket{{A}}_{\psi}\ket{\psi} + \braket{\Delta A}_{\psi} \ket{\psi^{\perp}_{A}},\label{Lev-2}
\end{align}
where $\ket{\psi^{\perp}_{A}}={\ket{\widetilde{\psi}_{A}^{\perp}}}/{\sqrt{\braket{\widetilde{\psi}_{A}^{\perp}|\widetilde{\psi}_{A}^{\perp}}}}$ and  $\sqrt{\braket{\widetilde{\psi}_{A}^{\perp}|\widetilde{\psi}_{A}^{\perp}}}=\sqrt{\braket{\psi|A^2|\psi}-\braket{\psi|A|\psi}^2}=\braket{\Delta A}_{\psi}$ and $\braket{A}_{\psi}=\braket{\psi|A|\psi}$.
\end{proof}

So, Eq. (\ref{RHUR-2}) can be interpreted as the unnormalized state vector ${A}\ket{\psi}$ which has two components and these are $\braket{{A}}_{\psi}$ along $\ket{\psi}$ and $\braket{\Delta A}_{\psi}$ along $\ket{\psi^{\perp}_{A}}$. Here we interpret   $\braket{\Delta A}_{\psi}$ as  disturbance of the state vector due to  the measurement of the operator $A$ or as the measurement error (or standard deviation) of that operator when the system is prepared in the state  $\ket{\psi}$. For instance, if we set up the system in one of the eigenstates of the observable $A$, then  from  Eq. (\ref{RHUR-2}), it can be seen that  the standard deviation of $A$ is zero.\\
\emph{\textbf{Information-theoretic interpretation.---}}
 From an information-theoretic approach, Eq. (\ref{RHUR-1}) can be written as
  \begin{align}
  \braket{\Delta{A}}_{\psi}=\sqrt{\sum_{i=1}^{d-1} {\left|{\braket{\psi_i^{\perp}|A|\psi}}\right|^2}},\label{RHUR-3}
 \end{align}
 where  \{$\ket{\psi}, \ket{\psi^{\perp}_{1}}, \ket{\psi^{\perp}_{2}}, \cdots, \ket{\psi^{\perp}_{d-1}}\}$ forms an orthonormal basis  such that  $I=\ket{\psi}\bra{\psi}+\sum_{i=1}^{d-1}{\ket{\psi_i^{\perp}}\bra{\psi_i^{\perp}}}$ and `$d$' is the dimension of the system. So, the origin of the non-zero standard deviation  $ \braket{\Delta{A}}_{\psi}$   in the standard quantum system  can also be thought of  due to the non-zero contributions of the unnormalized fidelities  \{$|\hspace{-1mm}\braket{\psi_i^{\perp}|A|\psi}\hspace{-1mm}|$\}$_{i=1}^{d-1}$ which can be viewed as the spread of the information of the observable $A$ along \{$\ket{\psi_i^{\perp}}$\}$_{i=1}^{d-1}$ directions. 
 \subsection{RHUR}
 The well known  RHUR    for two non-commuting operators $A$ and $B$  on a Hilbert space $\mathcal{H}$   when the system is  prepared in the state  $\ket{\psi}$ is given  by
\begin{align}
\braket{\Delta{A}}_{\psi}^{{2}}\braket{\Delta{B}}_{\psi}^{{2}}\geq\left[\frac{1}{2i}\braket{\psi|[A,B]|\psi}\right]^2,\label{RHUR-4}
\end{align}
where  $\braket{\Delta{A}}_{{\psi}}$ and  $\braket{\Delta{B}}_{{\psi}}$ are  the standard deviations of the operators $A$ and $B$, respectively, and $[A,B]=AB-BA$ is the commutator of $A$ and $B$. The derivation of Eq. (\ref{RHUR-4}) using the \emph{Aharonov–Vaidman identity} can be found in \cite{Lev-formula,Leifer-2023}. The stronger version is obtained  by adding  the ``Schr$\ddot{o}$dinger's term"  in Eq. (\ref{RHUR-4}) as 
\begin{align}
\braket{\Delta{A}}_{\psi}^{{2}}\braket{\Delta{B}}_{\psi}^{{2}}&\geq \left[\frac{1}{2i}\braket{\psi|[A,B]|\psi}\right]^2 \nonumber\\
&+ \left[\frac{1}{2}\braket{\psi|\{A,B\}|\psi} - \braket{A}_{\psi}\braket{B}_{\psi}\right]^2.\label{RHUR-5}
\end{align}\par
The {RHUR} is usually interpreted as the following: it puts bound on the sharp preparation  of a quantum state for two non-commuting observables. Hence, a quantum state in which the standard deviations of the two non-commuting observables are both zero cannot exist.
\section{Main results}\label{III}
The idea of the standard deviation or information dispersion (see preceding section) is a crucial component of the theory in a preparation-measurement situation. Pre- and post-selected systems are typical examples, therefore we define the standard deviation (uncertainty) of an observable and show that for such systems, there exist RHUR-like uncertainty relations for two non-commuting observables. \par
\subsection{Standard deviation in PPS system}  
\hspace{-3.6mm}\emph{\textbf{Geometric definition.---}}
It is well known that when the pre-selection and the post-selection are same, the PPS system becomes the standard quantum system (see introduction). The  following  proposition generalizes Eq. (\ref{RHUR-2}) for the PPS system.
\begin{proposition}
If a PPS system is in a pre-selected state $\ket{\psi}$ and post-selected state $\ket{\phi}$, then for a hermitian operator $A\in \mathcal{L}(\mathcal{H})$, we can decompose $A \ket{\psi}$ as 
\begin{align}
{A}\ket{\psi}=\braket{\phi |{A}|\psi}\ket{\phi} + \braket{\Delta A}^{\phi}_{\psi}  \ket{\phi^{\perp}_{A\psi}},\label{SDPPS-1}
\end{align}
where 
\begin{align}
\braket{\Delta A}^{\phi}_{\psi}=\sqrt{\braket{\psi|A^2|\psi} - |\braket{\phi|A|\psi}|^2} \label{SDPPS-2}
\end{align}
and  
$\ket{\phi^{\perp}_{A\psi}}=\frac{1}{\braket{\Delta A}_{\psi}^{\phi}}({A\ket{\psi}}-\braket{\phi|A|\psi}\ket{\phi})$, a normalized state vector which is orthogonal to $\ket{\phi}$.
\end{proposition}
\begin{proof}
 We assume that  $A\ket{\psi}$ and $\ket{\phi}$  are two non-orthogonal state vectors.  The unnormalized  state vector $\ket{\widetilde{\phi}^{\perp}_{A\psi}}\in\mathcal{H}$ which is  orthogonal to $\ket{\phi}$ is obtained using Gram-Schmidt orthonormalization process as 
\begin{equation}
 \ket{\widetilde{\phi}^{\perp}_{A\psi}} = {A}\ket{\psi} - \frac{\braket{\phi|({A}|\psi})}{\braket{\phi|\phi}}\ket{\phi}
 ={A}\ket{\psi} - \braket{\phi|A|\psi}\ket{\phi},\label{SDPPS-3}
\end{equation}
where $\braket{\phi|\phi}=1$ and after normalization, Eq. (\ref{SDPPS-3}) becomes
\begin{align}
{A}\ket{\psi}=\braket{\phi |{A}|\psi}\ket{\phi} + \braket{\Delta A}^{\phi}_{\psi}  \ket{\phi^{\perp}_{A\psi}},\nonumber
\end{align}
where $\ket{\phi^{\perp}_{A\psi}}={\ket{\widetilde{\phi}_{A\psi}^{\perp}}}/{\sqrt{\braket{\widetilde{\phi}_{A\psi}^{\perp}|\widetilde{\phi}_{A\psi}^{\perp}}}}$ and  $\sqrt{\braket{\widetilde{\phi}_{A\psi}^{\perp}|\widetilde{\phi}_{A\psi}^{\perp}}}=\sqrt{\braket{\psi|A^2|\psi}-|\braket{\phi|A|\psi}|^2}=\braket{\Delta A}^{\phi}_{\psi}$.
\end{proof}
\par
 To define the standard deviation of the observable $A$ in the PPS system, we now present an argument which is similar to the one used to describe the standard deviation of an observable in a standard quantum system.  So,  Eq. (\ref{SDPPS-1}) can be interpreted geometrically  as the unnormalized state vector ${A}\ket{\psi}$ which  has two components  $\braket{\phi|A|\psi}$ along the post-selection  $\ket{\phi}$ and $\braket{\Delta A}^{\phi}_{\psi}$ along $\ket{\phi^{\perp}_{A\psi}}$. \emph{{Here we define $\braket{\Delta A}^{\phi}_{\psi}$ as the {standard deviation} of the observable  $A$ when the system is pre-selected  in   $\ket{\psi}$ and post-selected in  $\ket{\phi}$}}.\par

The standard deviation $\braket{\Delta A}^{\phi}_{\psi}$  can be realized via the weak value of the observable $A$ as  
\begin{align}
 \braket{\Delta{A}}^{\phi}_{\psi}&=\sqrt{\braket{\psi|A^2|\psi} - |{\braket{A_w}^{\phi}_{\psi}}|^2|{\braket{\phi|\psi}}|^2}\label{SDPPS-4}\\
 &=\sqrt{\braket{\Delta{A}}^{2}_{\psi} +  {\braket{A}^2_{\psi}}  - |{\braket{A_w}^{\phi}_{\psi}}|^2|{\braket{\phi|\psi}}|^2},\label{SDPPS-5}
\end{align}  
 where $\braket{A_w}^{\phi}_{\psi}$ is the weak value of the observable $A$ defined in Eq. (\ref{intro-1}) and we have used Eq. (\ref{RHUR-1}) to derive Eq. (\ref{SDPPS-5}).  $|{\braket{\phi|\psi}}|^2$ is the  success probability of the post-selection $\ket{\phi}$. Eq. (\ref{SDPPS-4}) is no longer a valid expression  if pre- and post-selected states are orthogonal to one another because in this situation, weak value is not defined. Then, go back to Eq. (\ref{SDPPS-2}).  It should be noted that Eq. (\ref{SDPPS-2}) holds true whether the measurement is strong or weak.\\
\emph{\textbf{Information-theoretic definition.---}}
  Another expression of the standard deviation  $\braket{\Delta{A}}^{\phi}_{\psi}$ in the PPS system can be derived  by inserting an identity operator $I=\ket{\phi}\bra{\phi}+\sum_{i=1}^{d-1}{\ket{\phi_i^{\perp}}\bra{\phi_i^{\perp}}}$, {where \{$\ket{\phi}, \ket{\phi^{\perp}_{1}}, \ket{\phi^{\perp}_{2}}, \cdots, \ket{\phi^{\perp}_{d-1}} \}$} forms an orthonormal basis,  in the first term of the right hand side of Eq. (\ref{SDPPS-4}) as
 \begin{align}
  \braket{\Delta{A}}^{\phi}_{\psi}=\sqrt{\sum_{i=1}^{d-1} {\left|{\braket{A_w}^{\phi_i^{\perp}}_{\psi}}\right|^2|{\braket{\phi_i^{\perp}|\psi}}|^2}}.\label{SDPPS-6}
 \end{align}\par
 From an information-theoretic perspective, Eq. (\ref{SDPPS-6}) may now be understood as follows: non-zero standard deviation in the PPS system arises as a result of the non-zero contributions from the weak values  \{$\braket{A_w}^{\phi_i^{\perp}}_{\psi}$\}$_{i=1}^{d-1}$ along {the orthogonal post-selections} \{$\ket{\phi_i^{\perp}}$\}$_{i=1}^{d-1}$.  Note that two consecutive measurements are taken into account in a PPS system: the operator of interest $A$ and the projection operator  $\Pi_{\phi}=\ket{\phi}\bra{\phi}$ which corresponds to the post-selection $\ket{\phi}$.  {As a result, it is hard to tell whether or not $A$ has been measured when the weak value is zero. Because of this, it is crucial to have non-zero weak values which carry the information about  the observable $A$.    Null weak values have recently been given a useful interpretation \cite{Duprey-2017}: if a successful post-selection occurs with a null weak value, then  the property represented by the observable $A$ cannot be detected by the weakly coupled quantum pointer. In other words,  the pointer state remains unchanged  when the weak value is zero (see introduction section). Thus, one should anticipate that the standard deviation in the PPS system should be zero if we obtain null weak values for the post-selections \{$\ket{\phi_i^{\perp}}$\}$_{i=1}^{d-1}$,  that means the information about the observable $A$ is not dispersed throughout the post-selections \{$\ket{\phi_i^{\perp}}$\}$_{i=1}^{d-1}$.}\par
 {In addition to the standard deviation's geometrical and information-theoretical explanations (Eqs. (\ref{SDPPS-1}) and (\ref{SDPPS-6}), respectively) in the PPS system, we now study the minimum (zero) and maximum uncertainty (or standard deviation) which  provide additional insights to understand the standard deviation.}\par
  \emph{\textbf{Zero uncertainty.---}} The uncertainty $\braket{\Delta{A}}^{\phi}_{\psi}$ defined in Eq. (\ref{SDPPS-2}) in the PPS system is  zero if and only if 
\begin{align}
{A}\ket{\psi} &= \braket{{\phi_z}|A|\psi}\ket{\phi_z},\label{SDPPS-7}
\end{align}
or $\ket{\phi_z}\propto {A}\ket{\psi}$. We have used the notation $\ket{\phi_z}$ as the post-selection for which uncertainty  in PPS system becomes zero. The zero uncertainty in the PPS system can now be realised in the following way:  the weak value {$\braket{{A_w}}^{\phi_z}_{\psi}$} becomes non-zero i.e., $\frac{\braket{\psi|A^2|\psi}}{\braket{\psi|A|\psi}}\neq 0$  when we post-select the system to $\ket{\phi_z}$, and the weak values for all post-selections \{$\ket{{\phi_z}_i^{\perp}}$\}$_{i=1}^{d-1}$ orthogonal to $\ket{\phi_z}$ are zero. As a result, the right side of Eq. (\ref{SDPPS-6}) is reduced to zero. It should be noted that all post-selections orthogonal  to $\ket{\phi_z}$ are ``legitimate post-selections," meaning that their weak values are clearly specified. Equivalently, we can state that  the information about the observable $A$ is not dispersed along  the post-selections \{$\ket{{\phi_z}_i^{\perp}}$\}$_{i=1}^{d-1}$ as null weak values do not carry informations about the observable $A$ (according to the above information-theoretic definition). Hence, it is guaranteed that in a particular direction  there will be one and only one non-zero weak value of $A$ in a PPS system if and only if the condition (\ref{SDPPS-7}) is met.\par
\emph{Usefulness of zero uncertainty state:} In this paragraph we provide the following usefulness of the zero uncertainty post-selected state $\ket{\phi_z}$.
\par
  \textbf{1)} In a  parameter estimation  scenario, where the task is to obtain the precision limit in the estimation of interaction coefficient `g' in the interaction Hamiltonian $H=gA\otimes p$ (`$p$' is the pointer's momentum variable), Fisher information plays an important role whose   maximum value  is given by $F^{max}(g)=4\sigma^2\braket{\psi|A^2|\psi}$, where $\sigma$ is the standard deviation of initial distribution of the pointer {state and $\ket{\psi}$ is the initially prepared state of the system} \cite{Alves-2015}.  In an arbitrarily post-selected state $\ket{\phi}$, Fisher information is given by $F_{\phi}(g)=4\sigma^2|\hspace{-1mm}\braket{\phi|A|\psi}\hspace{-1mm}|^2 \leq F^{max}(g)$ \cite{Alves-2015}. Recently, {it was shown in the Ref. \cite{Shukur-2020} that the Fisher information can be expressed in terms of quasiprobability distribution of an arbitrary pre-selected state  when the system is a PPS system. Although the quasiprobabilities are in general complex and can take negative values as well, the Fisher information is always positive and real. The Fisher information of the pre-selected state with negative quasiprobability distribution  can surpass the usual quantum Fisher information (defined in standard quantum system).} Violation of such limit implies that the error which  occurs in estimating the unknown parameter can be reduced significantly  using the Fisher information  of the pre-selected state with negative quasiprobability distribution compared to  the usual  quantum Fisher information. One can immediately see using Eq. (\ref{SDPPS-2}) that $F_{\phi}(g)=4\sigma^2[\braket{\psi|A^2|\psi}- (\braket{\Delta{A}}^{\phi}_{\psi})^2]$. Now it is obvious that for the zero uncertainty post-selected state $\ket{\phi_z}$ as appeared in Eq. (\ref{SDPPS-7}), we have $F_{\phi_z}(g)=F^{max}(g)$.  Hence, to achieve the maximum Fisher information $F^{max}(g)$ in the PPS system, one must post-select  the system in $\ket{\phi_z}=A\ket{\psi}/\sqrt{\braket{\psi|A^2|\psi}}$ which corresponds to the zero uncertainty. \par
\textbf{2)} The post-selection $\ket{\phi_z}$  alone has the ability to provide the information (e.g., $\braket{\Delta A}_{\psi}$ and $\braket{A}_{\psi}$) about the observable $A$. Indeed by noting  that  $\braket{\psi|A^2|\psi}=p_z(\braket{A_w}^{\phi_z}_{\psi})^2$ and $\braket{\psi|A|\psi}=p_z\braket{A_w}^{\phi_z}_{\psi}$, we have $\braket{\Delta A}_{\psi}^2 =(1-p_z)p_z(\braket{A_w}^{\phi_z}_{\psi})^2$,  where $p_z=|\hspace{-1mm}\braket{\phi_z|\psi}\hspace{-1mm}|^2$ is the probability of obtaining the post-selection $\ket{\phi_z}=A\ket{\psi}/\sqrt{\braket{\psi|A^2|\psi}}$.\par
 {\emph{\textbf{Maximum uncertainty.---}}} To achieve the maximum  value of $\braket{\Delta{A}}^{\phi}_{\psi}$, the weak value $\braket{{A}_w}^{\phi}_{\psi}$ in Eq. (\ref{SDPPS-4})  has to be zero i.e., when the post-selection $\ket{\phi}$ is orthogonal to $A\ket{\psi}$ and hence max($\braket{\Delta{A}}^{\phi}_{\psi})=\sqrt{\braket{\psi|A^2|\psi}}$.  {Note that, in a preparation-measurement scenario, maximum measurement error is also found to be $\sqrt{\braket{\psi|A^2|\psi}}$ whether the measurement of the observable $A$ is performed in standard  system (see Eq. (\ref{RHUR-1}))  or while performing the best estimation the operator $A$ from the measurement of another  hermitian operator \cite{Hall-2004}}.\par

\subsection{Uncertainty relation in PPS system} \label{III B}

After defining the  standard deviation of an observable  in a PPS system, interpreting  it geometrically and informationally, and maintaining a parallel comparison and connection with the standard deviation in the standard  system, we are now in a position to provide an uncertainty relation in a PPS system for two incompatible observables.  One can formulate many different  types of uncertainty relations  in PPS systems (for example, \cite{Pati-Wu}), but our interpretation of an uncertainty relation in a PPS system is based on the standard deviation  defined in Eq. (\ref{SDPPS-2}) or (\ref{SDPPS-4}). Since the weak value of the  observable $A$ in the standard deviation Eq. (\ref{SDPPS-4}) in the PPS system replaces the average value of the same  observable $A$ in the standard deviation Eq. (\ref{RHUR-1}) in standard quantum system, it is not surprising that the mathematical expression of the uncertainty relation in the PPS system is a modified version of the RHUR (\ref{RHUR-4}), where the average values of the incompatible observables $A$ and $B$ in Eq. (\ref{RHUR-4})  will be replaced by the weak values of the respective observables when the system is pre-selected in $\ket{\psi}$ and post-selected in $\ket{\phi}$. The explicit form of the uncertainty relation in the PPS system is provided in the following theorem.
\\
\begin{theorem}
    Let $A$, $B$ $\in \mathcal{L}(\mathcal{H})$  be two non-commuting hermitian operators which are measured in  the PPS system of our interest with $\ket{\psi}$ and $\ket{\phi}$ being pre- and post-selected states, respectively, then the product of their standard deviations satisfies 
 \begin{align}
 \left(\hspace{-0mm}\braket{\Delta{A}}^{\phi}_{\psi}\hspace{-0mm}\right)^2\hspace{-1mm}\left(\hspace{-0mm}\braket{\Delta{B}}^{\phi}_{\psi}\hspace{-0mm}\right)^2 \geq \left[\frac{1}{2i}\bra{\psi}\hspace{-1mm}[{A},{B}]\hspace{-1mm}\ket{\psi}  - \mathcal{I}m \left(W_{AB}\right) \right]^2\hspace{-1mm},\label{URPPS-1}
  \end{align} 
 where 
$W_{AB}\hspace{-1mm}=\hspace{-1mm}\braket{\psi|A|\phi}\hspace{-1mm}\braket{\phi|B|\psi}\hspace{-1mm}=\hspace{-1mm}\left(\hspace{-0mm}\braket{A_w}^{\phi}_{\psi}\hspace{-0mm}\right)^{*} \hspace{-2mm}\braket{B_w}^{\phi}_{\psi}|{\braket{\phi|\psi}}|^2$ (using the definition of the  weak value defined in Eq. (\ref{intro-1})).
\end{theorem}
 \begin{proof}
     Cauchy-Schwarz inequality for two unnormalized state vectors $\ket{\widetilde{\phi}^{\perp}_{A\psi}}$ and $\ket{\widetilde{\phi}^{\perp}_{B\psi}}$ in $\mathcal{H}$  becomes
 \begin{align}
 {\braket{\widetilde{\phi}^{\perp}_{A\psi}|\widetilde{\phi}^{\perp}_{A\psi}}} {\braket{\widetilde{\phi}^{\perp}_{B\psi}|\widetilde{\phi}^{\perp}_{B\psi}}} \geq \left|{{\braket{\widetilde{\phi}^{\perp}_{A\psi}|\widetilde{\phi}^{\perp}_{B\psi}}}}\right|^2,\label{URPPS-2}
  \end{align} 
  Now, as $\left|{{\braket{\widetilde{\phi}^{\perp}_{A\psi}|\widetilde{\phi}^{\perp}_{B\psi}}}}\right|^2=[\mathcal{R}e({{\braket{\widetilde{\phi}^{\perp}_{A\psi}|\widetilde{\phi}^{\perp}_{B\psi}}}})]^2+[\mathcal{I}m({{\braket{\widetilde{\phi}^{\perp}_{A\psi}|\widetilde{\phi}^{\perp}_{B\psi}}}})]^2$ and hence 
 \begin{align}
 \left|{{\braket{\widetilde{\phi}^{\perp}_{A\psi}|\widetilde{\phi}^{\perp}_{B\psi}}}}\right|^2 \geq  [\mathcal{I}m({{\braket{\widetilde{\phi}^{\perp}_{A\psi}|\widetilde{\phi}^{\perp}_{B\psi}}}})]^2,\label{URPPS-3}
  \end{align} 
  where $\mathcal{I}m({{\braket{\widetilde{\phi}^{\perp}_{A\psi}|\widetilde{\phi}^{\perp}_{B\psi}}}})=\frac{1}{2i}\left( {\braket{\widetilde{\phi}^{\perp}_{A\psi}|\widetilde{\phi}^{\perp}_{B\psi}}} - {\braket{\widetilde{\phi}^{\perp}_{B\psi}|\widetilde{\phi}^{\perp}_{A\psi}}}  \right)$ and  $\mathcal{R}e({{\braket{\widetilde{\phi}^{\perp}_{A\psi}|\widetilde{\phi}^{\perp}_{B\psi}}}})=\frac{1}{2}(\braket{\widetilde{\phi}^{\perp}_{A\psi}|\widetilde{\phi}^{\perp}_{B\psi}} + {\braket{\widetilde{\phi}^{\perp}_{B\psi}|\widetilde{\phi}^{\perp}_{A\psi}}})$. Now put  $\ket{\widetilde{\phi}^{\perp}_{A\psi}}={A}\ket{\psi} - \braket{\phi|A|\psi}\ket{\phi}$ defined in  Eq. (\ref{SDPPS-3}) for operator $A$ and similarly for operator $B$ also, then  we have 
 \begin{align}
 {\braket{\widetilde{\phi}^{\perp}_{A\psi}|\widetilde{\phi}^{\perp}_{B\psi}}}=\braket{\psi|AB|\psi}-\braket{\psi|A|\phi}\braket{\phi|B|\psi}.\label{URPPS-4}
 \end{align} 
 Note that ${\braket{\widetilde{\phi}^{\perp}_{A\psi}|\widetilde{\phi}^{\perp}_{A\psi}}}=\left(\braket{\Delta{A}}^{\phi}_{\psi}\right)^2$ is  square of the standard deviation of the observable $A$ in the PPS system defined in Eq. (\ref{SDPPS-2}) and similarly $\braket{\widetilde{\phi}^{\perp}_{B\psi}|\widetilde{\phi}^{\perp}_{B\psi}}=\left(\braket{\Delta{B}}^{\phi}_{\psi}\right)^2$ is square of the standard deviation of the observable $B$ in the PPS system. Finally, putting  these values and  using Eqs. (\ref{URPPS-3}) and  (\ref{URPPS-4}) in Eq. (\ref{URPPS-2}),  it becomes  Eq. (\ref{URPPS-1}).
 \end{proof}
{Eq.  (\ref{URPPS-1}) is always true for any  strong PPS systems \cite{ABL1964, Aharonov-Vaidman-2008, Kofman-Nori} or weak PPS systems \cite{AAV-1988}}.  For  weak PPS measurements \cite{AAV-1988}, $W_{AB}$ is expressed in terms of weak values of   both the observables. If the pre- and post-selected states are same i.e.,  $\ket{\phi}=\ket{\psi}$, then one  gets back the  RHUR  (\ref{RHUR-4}) as argued before.  Eq. (\ref{URPPS-1}) with ``Schr$\ddot{o}$dinger's term" becomes 
 \begin{align}
 \left(\hspace{-0mm}\braket{\Delta{A}}^{\phi}_{\psi}\hspace{-0mm}\right)^2\hspace{-1mm}\left(\hspace{-0mm}\braket{\Delta{B}}^{\phi}_{\psi}\hspace{-0mm}\right)^2 \geq &\left[\frac{1}{2i}\bra{\psi}\hspace{-1mm}[{A},{B}]\hspace{-1mm}\ket{\psi}  - \mathcal{I}m \left(W_{AB}\right) \right]^2 \nonumber \\
  + &  \left[\frac{1}{2}\braket{\psi|\{A,B\}|\psi} - \mathcal{R}e\left(W_{AB}\right)\,\right]^2\hspace{-1mm}.\label{URPPS-5}
  \end{align} \par
The uncertainty relation (\ref{URPPS-1}) can be interpreted in the same way as we did for the RHUR   (\ref{RHUR-4}). That is,  it bounds the sharp preparation of the  pair  for pre- and post-selections ($\ket{\psi}$, $\ket{\phi}$) for  two non-commuting observables. The lower bound contains an additional term $\mathcal{I}m(W_{AB})$ compared to the  RHUR (\ref{RHUR-4}). So even if $[A,B] \neq 0$, the bound on the right hand side of Eq. (\ref{URPPS-1}) can become zero implying the possibility of both the standard deviations being zero implying further the possibility of sharp preparation of a pair of pre- and post-selected states.   Below, we provide  the necessary and sufficient  condition for  such case (see \emph{Observation 2}). Recently, the authors of the references \cite{Bao-2020,Molmer-2021} confirmed that in a PPS system using the ABL- rule \cite{ABL1964, Aharonov-Vaidman-2008, Kofman-Nori}, it is possible to go beyond the standard lower bound in  the RHUR for position and momentum observables. Not exactly, but a similar property i.e., achieving arbitrary small lower bound  (which depends on the pre- and post-selections) of the product of standard deviations of two non-commuting observables  in a PPS system is possible in the relations (\ref{URPPS-1}). We now explore two peculiar characteristics of the uncertainty relations (\ref{URPPS-1}) and  (\ref{URPPS-5}) that cannot be observed in standard quantum systems.   \par
\emph{\textbf{Observation 1.}} { If the lower bound in an uncertainty relation in any quantum system is non-zero, then we say that two incompatible observables disturb each others' measurement results. Now consider the following case.  If  $\ket{\psi}$ is a common eigenstate of both  ${A}$ and ${B}$, then  $\braket{\Delta{A}}_{\psi}=0$, $\braket{\Delta{B}}_{\psi}=0$ implying that the measurement of one doesn't disturbs the outcome of the other.  Surprisingly, this property doesn't hold in the PPS system. Note that, even if $\ket{\psi}$ is a common eigenstate of both  ${A}$ and ${B}$, the lower bound of the relation (\ref{URPPS-5})  doesn't become zero  for specific post-selections which implies $\braket{\Delta{A}}^{\phi}_{\psi}\neq 0$, $\braket{\Delta{B}}^{\phi}_{\psi}\neq 0$. Hence we can say that  the measurement of ${A}$ is invariably disturbed by the measurement of ${B}$  or vice versa in a PPS system. In the Ref.  \cite{Vaidman-1993}, Vaidman demonstrated the same property in a PPS system using the ABL- rule}.\par
\emph{\textbf{Observation 2.}} With two non-commuting observables in the standard quantum system, sharp preparation of a quantum state is impossible. Or equivalently, for  an  initially prepared state $\ket{\psi}$, it is  impossible to have $\braket{\Delta{A}}_{\psi}=0$, $\braket{\Delta{B}}_{\psi}=0$ if $[A,B]\neq 0$. But in the PPS system, we can prepare \textit{any} quantum state $\ket{\psi}$ which can give $\braket{\Delta{A}}^{\phi}_{\psi}=0$, $\braket{\Delta{B}}^{\phi}_{\psi}=0$ for a specific choice of post-selection   implying sharp preparation of $\ket{\psi}$ for  non-commuting observables $A$ and $B$. It is easy to show that  both the uncertainties  $\braket{\Delta{A}}^{\phi}_{\psi}$ and  $\braket{\Delta{B}}^{\phi}_{\psi}$ are  zero for the common post-selection $\ket{{\phi}_z}$ if and only if
\begin{equation} 
\begin{aligned}
 \ket{{\phi}_z} \propto A\ket{\psi},\hspace{1cm}
\ket{{\phi}_z} \propto B\ket{\psi}.\nonumber
\end{aligned}
\end{equation}
After the normalization, we find the common post-selection condition
\begin{align}
\ket{{\phi}_z}= \frac{A\ket{\psi}}{\sqrt{\braket{\psi|A^2|\psi}}}=\frac{B\ket{\psi}}{\sqrt{\braket{\psi|B^2|\psi}}},\label{URPPS-6}
\end{align}
upto some phase factors.\par
\emph{Example.} Now, consider an example of two non-commuting observables $A=\frac{1}{\sqrt{2}}(I+\sigma_x)$ and $B=\frac{1}{\sqrt{2}}(\sigma_z+\sigma_x)$ with the initially prepared state  $\ket{0}$. With these specific choices, it is possible to show that condition (\ref{URPPS-6}) is satisfied. Recall that in order to conduct the experiment using weak values, the average values of the observables must not be zero; for this reason, we did not take into account the Pauli observables $\sigma_x$ and $\sigma_y$ with initially prepared state $\ket{0}$. Nonetheless, if one does not adhere to weak values, this example is  still  true. So, the common  post-selection  for this case is  $\ket{\phi_{e}}=(\ket{0}+\ket{1})/\sqrt{2}$  and hence both the uncertainties $\braket{\Delta{A}}^{\phi_{e}}_{0}$ and $\braket{\Delta{B}}^{\phi_{e}}_{0}$ of the non-commuting observables $A$ and $B$, respectively are zero  for the given initially prepared state $\ket{0}$ and the conditioned   post-selection  $\ket{\phi_{e}}$ in Eq. (\ref{URPPS-6}).  In the PPS system, it is now feasible to do the hitherto impossibly difficult task of jointly sharply preparing a quantum state for two non-commuting observables. \par
 The aforementioned \emph{Observations 1} and \emph{2}   demonstrate that PPS systems are capable of being even stranger than their well-known unusual results e.g., quantum Cheshire Cats \cite{Cheshire-2013},  measurement of a component of a spin 1/2  particle which can reach $100\hbar$ \cite{AAV-1988}, etc. \par
 \emph{\textbf{Comments.---}} The characteristics of the uncertainty relations (\ref{URPPS-1}) and (\ref{URPPS-5}) in PPS systems  as compared to the RHUR  (\ref{RHUR-4}) and Eq. (\ref{RHUR-5})  are substantially altered by the post-selections.   These uncertainty inequalities (\ref{URPPS-1}) and (\ref{URPPS-5}) will undoubtedly have applications like the RHUR for quantum foundations, information and technologies. For instances, $(i)$ they can be used for information extraction using commuting observables because the inequalities do not become trivial for particular choices of post-selections, $(ii)$ one can obtain a series of uncertainty inequalities by changing the post-selections and that is advantageous for practical purposes (see stronger uncertainty relations in Sec. \ref{IV}), $(iii)$ existing applications of uncertainty relations (\ref{RHUR-4}) and (\ref{RHUR-5}) in standard systems, such as entanglement detection \cite{Hofmann-2003}, quantum metrology \cite{Giovannetti-2006, Pezze-2018}, etc., can be revisited using uncertainty relations (\ref{URPPS-1}) and (\ref{URPPS-5}) in the PPS systems, ($iv$) PPS system based spin squeezing: spin-squeezed states are a class of states having squeezed spin variance along a certain direction, at the cost of anti-squeezed variance along an orthogonal direction. This is done by using the RHUR  (\ref{RHUR-4}) in the standard quantum system \cite{Ueda-1993,Wineland-1992,Wineland-1994,Ma-2011}. Such analysis can be reintroduced  in the light of PPS systems. {As there is no unique definition of spin squeezing in the standard quantum systems, it is, by means of Eq. (\ref{URPPS-1}), also possible to define the spin squeezing non uniquely in the PPS system}. A very careful analysis is required to see whether there exists some states in the PPS systems for which  $\braket{\Delta{A}}^{\phi}_{\psi}=\braket{\Delta{B}}^{\phi}_{\psi}$ and inequality (\ref{URPPS-1} is saturated similar to coherent spin states  in  the standard quantum systems.   \par
 \subsubsection{Intelligent pre- and post-selected states}
 {In the standard quantum system, the states for which the equality condition holds in the RHUR (\ref{RHUR-4}) are known as  intelligent states  or minimum-uncertainty states \cite{Aragone-1974, Chalbaud-1976, Trifonov-1994}. Minimum uncertainty  states have been proposed to improve the accuracy of phase measurement in quantum interferometer \cite{Mlodinow-1993}. Minimum-uncertainty states   in the PPS systems  can also be defined based on the  following condition.}\par
  One can  find the condition for which the inequality (\ref{URPPS-1}) saturates (see Appendix \ref{A}) is given by 
 \begin{equation}
 \hspace{-2mm}{A}\hspace{-1mm}\ket{\psi}\hspace{-0mm} - \hspace{-0mm}\braket{\phi|A|\psi}\hspace{-1mm}\ket{\phi}=\pm i\frac{\braket{\Delta{A}}^{\phi}_{\psi}}{\braket{\Delta{B}}^{\phi}_{\psi}}\left({B}\hspace{-1mm}\ket{\psi} \hspace{-0mm}- \hspace{-0mm}\braket{\phi|B|\psi}\hspace{-1mm}\ket{\phi}\right).\label{URPPS-7}
\end{equation}
{If the sign of `$i$' on the RHS of  Eq. (\ref{URPPS-7}) is taken to be positive (negative) when the observable $A$  appears on the LHS of the Eq. (\ref{URPPS-7}), then the sign of `$i$' on the RHS of  Eq. (\ref{URPPS-7}) is taken be  negative (positive) when the observable $B$  appears on the LHS of the Eq. (\ref{URPPS-7}).}  So, the pre- and post-selected states which satisfy the  condition (\ref{URPPS-7}, can be referred as  the ``intelligent pre- and post-selected states". {For the given pre-selection and observables in Eq. (\ref{URPPS-7}), one can find the post-selection which will make the Eq. (\ref{URPPS-1}) the most tight i.e., equality. 
 \par

\subsection{Uncertainty equality in  PPS system}
Recently, in the reference \cite{Yao-2015},  the authors have shown that there exist   variance-based uncertainty equalities from which a series of uncertainty inequalities with hierarchical structure can be obtained. It was shown that stronger uncertainty relation given by Maccone and Pati \cite{Maccone-Pati} is a special case of these  uncertainty inequalities. Here we show such uncertainty equalities in the PPS systems.  We provide  interpretation of the uncertainty inequalities derived from the uncertainty equalities. Further, in application section \ref{IV}, we use  uncertainty equalities in PPS systems to obtain stronger uncertainty relations and state dependent tighter  uncertainty relations.
\begin{theorem}
   The product of standard deviations of two non-commuting  hermitian operators $A$, $B$ $\in \mathcal{L}(\mathcal{H})$  in  a PPS system  with pre- and post-selected states  $\ket{\psi}$ and $\ket{\phi}$, respectively  satisfies 
\begin{align}
\braket{\Delta{A}}^{\phi}_{\psi}\braket{\Delta{B}}^{\phi}_{\psi}\hspace{-1mm}=\frac{\mp\left(\frac{1}{2i}\braket{\psi|[A,B]|\psi}-\mathcal{I}m(W_{AB})\right)}{1-\frac{1}{2}\sum_{k=1}^{d-1}{\big|\braket{\psi|\frac{A}{\braket{\Delta{A}}^{\phi}_{\psi}}\pm i\frac{B}{\braket{\Delta{B}}^{\phi}_{\psi}}|\phi^{\perp}_{k}}\big|^2}},\label{UEPPS-1}
\end{align}
where we have assumed that $\braket{\Delta{A}}^{\phi}_{\psi}$ and $\braket{\Delta{B}}^{\phi}_{\psi}$ are non-zero, and the sign should be considered such that the numerator is always  real and positive. Here  \{$\ket{\phi}, \ket{\phi_k^{\perp}}^{d-1}_{k=1}$\}  is an complete orthonormal  basis  in the $d$-dimensional Hilbert space.
\end{theorem}  
\begin{proof}
    Consider  an orthonormal complete basis  \{$\ket{\phi}, \ket{\phi_k^{\perp}}^{d-1}_{k=1}$\} in the $d$-dimensional Hilbert space $\mathcal{H}$.  Now, define the projection operator $\Pi=I-\ket{\phi}\bra{\phi}$ and the unnormalized  state vector  $\ket{\xi^{\pm}}=\big(\frac{A}{\braket{\Delta{A}}^{\phi}_{\psi}}\pm i\frac{B}{\braket{\Delta{B}}^{\phi}_{\psi}}\big)\ket{\psi}$. Then we have  the following identity
\begin{align}
\hspace{-6pt}\braket{\xi^{\mp}|\Pi|\xi^{\mp}}&\hspace{-1mm}=\hspace{-1mm}\braket{\xi^{\mp}|\xi^{\mp}} - \braket{\xi^{\mp}|\phi}\braket{\phi|\xi^{\mp}}\nonumber\\
&=\hspace{-1mm}\left\{\hspace{-1mm}\frac{\braket{\psi|A^2|\psi}}{\left(\hspace{-1mm}\braket{\Delta{A}}^{\phi}_{\psi}\hspace{-1mm}\right)^2}+\frac{\braket{\psi|B^2|\psi}}{\left(\hspace{-1mm}\braket{\Delta{B}}^{\phi}_{\psi}\hspace{-1mm}\right)^2}\mp \frac{i\braket{\psi|[A,B]|\psi}}{\braket{\Delta{A}}^{\phi}_{\psi}\braket{\Delta{B}}^{\phi}_{\psi}}\hspace{-1mm}\right\}\nonumber\\
&-\hspace{-1mm}\left\{\hspace{-1mm}\frac{|{\braket{\phi|A|\psi}}|^2}{\left(\hspace{-1mm}\braket{\Delta{A}}^{\phi}_{\psi}\hspace{-1mm}\right)^2}+\frac{|{\braket{\phi|B|\psi}}|^2}{\left(\hspace{-1mm}\braket{\Delta{B}}^{\phi}_{\psi}\hspace{-1mm}\right)^2}\pm \frac{2\mathcal{I}m(W_{AB})}{\braket{\Delta{A}}^{\phi}_{\psi}\hspace{-1mm}\braket{\Delta{B}}^{\phi}_{\psi}}\hspace{-1mm}\right\}\nonumber\\
&=2\pm 2\frac{\left(\frac{1}{2i}\braket{\psi|[A,B]|\psi}-\mathcal{I}m(W_{AB})\right)}{\braket{\Delta{A}}^{\phi}_{\psi}\braket{\Delta{B}}^{\phi}_{\psi}},\label{UEPPS-2}
\end{align}
where we have used Eq. (\ref{SDPPS-2}) and $W_{AB}=\braket{\psi|A|\phi}\braket{\phi|B|\psi}$. Now, we use another expression of  $\Pi=\sum_{k=1}^{d-1}{\ket{\phi^{\perp}_{k}}\bra{\phi^{\perp}_{k}}}$ to calculate the same identity
\begin{align}
\braket{\xi^{\mp}|\Pi|\xi^{\mp}}&=\sum_{k=1}^{d-1}{\Big|\braket{\psi|\frac{A}{\braket{\Delta{A}}^{\phi}_{\psi}}\pm i\frac{B}{\braket{\Delta{B}}^{\phi}_{\psi}}|\phi^{\perp}_{k}}\Big|^2}.\label{UEPPS-3}
\end{align}
So, from the  Eqs. (\ref{UEPPS-2}) and (\ref{UEPPS-3}), we obtain the uncertainty equality (\ref{UEPPS-1}) in the PPS system. 
\end{proof}
\begin{theorem}
   The sum of the variances of two non-commuting  hermitian operators $A$, $B$ $\in \mathcal{L}(\mathcal{H})$  in  a PPS system  with pre- and post-selected states  $\ket{\psi}$ and $\ket{\phi}$, respectively  satisfies 
\begin{align}
\left(\hspace{-1mm}\braket{\Delta{A}}^{\phi}_{\psi}\right)^2\hspace{-1mm} + \left(\hspace{-1mm}\braket{\Delta{B}}^{\phi}_{\psi}\right)^2
\hspace{-1mm}=&\pm\Big(i\braket{\psi|[A,B]|\psi}
\hspace{-1mm}-2 \mathcal{I}m(W_{AB})\hspace{-1mm}\Big)\nonumber\\
&+\sum_{k=1}^{d-1}|\hspace{-1mm}\braket{\phi_k^{\perp}|(A\mp iB)|\psi}\hspace{-1mm}|^2.\label{UEPPS-4}
\end{align}
Here, the `$\pm$' sign is taken  suitably such that the first  term  in right side  is always positive.
\end{theorem}  

 \begin{proof}
    Consider  an orthonormal complete basis  \{$\ket{\phi}, \ket{\phi_k^{\perp}}^{d-1}_{k=1}$\} in the $d$-dimensional Hilbert space $\mathcal{H}$ and hence $I-\ketbra{\phi}{\phi}=\sum_{k=1}^{d-1}{\ket{\phi^{\perp}_{k}}\bra{\phi^{\perp}_{k}}}$. By equating the following two 
\begin{align}
&Tr\Big((A\mp iB)\ketbra{\psi}{\psi}(A\pm iB)(I-\ket{\phi}\bra{\phi})\Big),\nonumber\\
&Tr\Big((A\mp iB)\ketbra{\psi}{\psi}(A\pm iB)(\sum_{k=1}^{d-1}\ket{\phi_k^{\perp}}\bra{\phi_k^{\perp}})\Big)\nonumber,
\end{align}
we have Eq. (\ref{UEPPS-4}).
 \end{proof}

 \begin{table*}[tbp]
\caption{Comparison of different properties between standard quantum systems and PPS systems.}
\label{Table}
\begin{center}
\begin{tabular}{c|c|c|c}
\hline\hline
Properties  & Standard  quantum systems & Pre- and post-selected  systems \\ \hline\hline
\parbox{2cm} {Standard deviation} & 
\parbox{5.6cm} { \ \ \ \ \ \ \ \
\ \ \ \ \ \ \ \ \ \ \ \ \ \ \ \ \ \ \ \ \ \ \ \ \ \ \ \ \ \ \ \ \ \
\ \ \ \ \ \ \ \ \ \ \ \ \ \ \ \ \ \ \ \ \ \ \ \ \ \ \ \ \ \ \ $\braket{\Delta{A}}_{\psi}= \left(\braket{\psi|A^2|\psi} - \braket{\psi|A|\psi}^2\right)^{1/2}$ \ \ \ \ \ \ \ \ \ \ \ \ \ \ \ \ \ \ \ \ \ \ \ \ \ \ \ \ \ \ \ \
\ \ \ \ \ \ } & 
\parbox{8.6cm} {  $\braket{\Delta{A}}^{\phi}_{\psi}=\left(\braket{\psi|A^2|\psi} - |{\braket{A_w}^{\phi}_{\psi}}|^2|{\braket{\phi|\psi}}|^2\right)^{1/2}$  \ \ \ \ \ \ \ \ \ \ \ \ \
\ \ \ \ \ \ \ \ \ \ \ \ \ \ \ \ \ \ \ \ \ \ \ \ \ } \\ \hline

\parbox{2cm} {Zero standard deviation} & 
\parbox{5.6cm} {\ \ \ \ \ \ \ \
\ \ \ \ \ \ \ \ \ \ \ \ \ \ \ \ \ \ \ \ \ \ \ \ \ \ \ \ \ \ \ \ \ \
\ \ \ \ \ \ \ \ \ \ \ \ \ \ \ \ \ \ \ \ \  Only if  $\ket{\psi}$ is an eigenstate of $A$  i.e., $\ket{\psi}\propto A\ket{\psi}$}  \ \ \ \ \
\ \ \ \ \ \ & 
\parbox{8.6cm} { \ \ \ \ \ \ \ \ \ \ \ \ \  Only if  $ \ket{\phi} \propto A\ket{\psi}$}  \\ \hline

\parbox{2cm} {Uncertainty relation} & 
\parbox{5.6cm} { \ \ \ \ \ \ \ \
\ \ \ \ \ \ \ \ \ \ \ \ \ \ \ \ \ \ \ \ \ \ \ \ \ \ \ \ \ \ \ \ \ \
\ $ \braket{\Delta{A}}_{\psi}^{{2}}\braket{\Delta{B}}_{\psi}^{{2}}\geq\left[\frac{1}{2i}\braket{\psi|[A,B]|\psi}\right]^2$ \ \
\ \ \ \ \ \ \ \ \ \ \ \ \ \ \ \ \ \ \ \ \ \ \ \ \ \ \ \ \ \ \ \ \ \ \ \ \ } & 
\parbox{8.6cm} { \ \ \ \ \ \ \ \   $\left(\braket{\Delta{A}}^{\phi}_{\psi}\right)^2\left(\braket{\Delta{B}}^{\phi}_{\psi}\right)^2 \hspace{-1mm}\geq \Big[\frac{1}{2i}\bra{\psi}\hspace{-1mm}[{A},{B}]\hspace{-1mm}\ket{\psi}  - \mathcal{I}m \left(W_{AB}\right) \Big]^2$  \ \ \ \ \ \ \ \ } \\ \hline

\parbox{2cm} {Joint sharp preparation} & 
\parbox{5.6cm} {If $\ket{\psi}$ is the eigenstate of  both $A$ and $B$} & 
\parbox{8.6cm} { If $\ket{\phi}=\frac{A\ket{\psi}}{\sqrt{\braket{\psi|A^2|\psi}}}=\frac{B\ket{\psi}}{\sqrt{\braket{\psi|B^2|\psi}}}$,  \hspace{0.5mm}  up to some phase factors} \\ \hline
\end{tabular}%
\end{center}
\end{table*}

An inequality can be obtained by discarding some of the terms in the summation corresponding to `$k$'  or all the terms  except one term  in  Eq. (\ref{UEPPS-1}) or Eq. (\ref{UEPPS-4}).  It is also possible to obtain an  arbitrarily tight inequality  by discarding the minimum valued  term  inside the summation   in the denominator of Eq. (\ref{UEPPS-1}) for a particular value of `$k$'. . Note that we have to optimize the minimum $\big|\braket{\psi|\frac{A}{\braket{\Delta{A}}^{\phi}_{\psi}}\pm i\frac{B}{\braket{\Delta{B}}^{\phi}_{\psi}}|\phi^{\perp}_{k}}\big|^2$ over all possible choice of basis \{$\ket{\phi_k^{\perp}}^{d-1}_{k=1}$\} in the  subspace  orthogonal to $\ket{\phi}$.\par
 In an experiment, let's assume that a few post-selected states from \{$\ket{\phi_k^{\perp}}\}^{d-1}_{k=1}$  are not detected by the detector because of certain technical difficulties. Using such imprecise experimental data, one may still be able to obtain an uncertainty relation. In that case,  the terms corresponding to the unregistered post-selections in Eq. (\ref{UEPPS-1}) or Eq. (\ref{UEPPS-4}) are to be eliminated.\par
 \subsection{Uncertainty relation for mixed pre-selection in PPS system}
 So far, we have only considered  the  pre-selected state to be  pure in a PPS system.  Let us now generalize the definition of the  standard deviation   and derive  the uncertainty relations  for  mixed pre-selected state in the PPS system.  A direct generalization of the standard deviation defined in  Eq. (\ref{SDPPS-2})  is given by
{ \begin{align}
\braket{\Delta A}^{\phi}_{\rho}=\sqrt{Tr(A^2\rho) - \braket{\phi|A\rho A|\phi}}.\label{MixedPPS-1}
\end{align}
See Ref. \cite{Explanation-2} for the motivation for calling $\braket{\Delta A}^{\phi}_{\rho}$ as a  standard deviation when the pre-selection is a mixed state. If $\rho=\sum_{i}p_i\ket{\psi_i}\bra{\psi_i}$, where $\sum_{i}p_i=1$, then the variance is given by}
 \begin{align}
\left(\braket{\Delta A}^{\phi}_{\rho}\right)^2=\sum_{i}p_i\left(\hspace{-1mm}\braket{\Delta{A}}^{\phi}_{\psi_i}\hspace{-1mm}\right)^2 .\label{MixedPPS-2}
\end{align}\par
\textit{Eq. (\ref{MixedPPS-2})  demonstrates the intriguing fact that, the variance of $A$ in PPS system i.e.,   $\braket{Var A}_{\rho}^{\phi}=(\braket{\Delta A}^{\phi}_{\rho})^2$  {respects} classical mixing of  quantum states}. Mathematically, classical mixing of  quantum states are represented by a density operator. By taking advantage of this property, one can obtain a purely quantum uncertainty relation when the pre-selection $\rho$ is a mixed state  (see Sec. \ref{IV}). It may be noted here that in  standard quantum systems, the variance  $Var A=\braket{\Delta A}^2_{\rho}$  increases, in general, under  the classical mixing of  quantum states.\par
To realize the standard deviation  in PPS system via weak value for mixed pre-selected state{, a generalization of  the standard deviation $ \braket{\Delta{A}}^{\phi}_{\psi}$ given in  Eq. (\ref{SDPPS-4}), can be \emph{defined} as}
 \begin{align}
\braket{\Delta A_w}^{\phi}_{\rho}=\sqrt{Tr(A^2\rho) -|\braket{A_w}^{\phi}_{\rho}|^2\braket{\phi|\rho|\phi}},\label{MixedPPS-3}
\end{align}
where $\braket{A_w}^{\phi}_{\rho}=\frac{\braket{\phi|A\rho|\phi}}{\braket{\phi|\rho |\phi}}$ is the weak value of the operator $A$ when the pre- and post-selections are $\rho$ and $\ket{\phi}$, respectively. {Now,  $\braket{Var A_w}_{\rho}^{\phi}:=(\braket{\Delta A_w}^{\phi}_{\rho})^2$ can be viewed as a variance-like quantity (henceforth called as \emph{generalized variance})  of $A$ involving weak value and, it is,  in general,  different from the variance $\braket{Var A}_{\rho}^{\phi}=(\braket{\Delta A}^{\phi}_{\rho})^2$,  and $\braket{Var A_w}_{\rho}^{\phi}$ is  always non-decreasing  under the classical mixing of quantum states. This property  is certified by the following {proposition}.}
\begin{proposition}
    {The generalized variance $\braket{Var A_w}_{\rho}^{\phi}$  is lower bounded by the variance $\braket{Var A}_{\rho}^{\phi}$, that is 
    \begin{align}
\braket{Var A_w}_{\rho}^{\phi}\geq \sum_{i}p_i\left(\hspace{-1mm}\braket{\Delta{A}}^{\phi}_{\psi_i}\hspace{-1mm}\right)^2 =\braket{Var A}_{\rho}^{\phi},\label{MixedPPS-4}
 \end{align} 
  where  equality holds if the pre-selection  $\rho$ is  a pure state.}
\end{proposition}
\begin{proof}
Let $\rho=\sum_{i}p_i\ket{\psi_i}\bra{\psi_i}$, then using  Eq. (\ref{MixedPPS-3}) (after using the definition of the weak value for mixed pre-selection), we have 
\begin{align}
&{\braket{Var A_w}_{\rho}^{\phi}}\nonumber\\
&=\left(\braket{\Delta A_w}^{\phi}_{\rho}\right)^2\nonumber\\
&=Tr(A^2\rho) -\frac{|\braket{\phi|A\rho|\phi}|^2}{\braket{\phi|\rho|\phi}}\nonumber\\
&=\sum_{i}p_i\braket{\psi_i|A^2|\psi_i}-\frac{|\sum_{i}\sqrt{p_i}\braket{\phi|A|\psi_i}\sqrt{p_i}\braket{\psi_i|\phi}|^2}{\braket{\phi|\rho|\phi}}\nonumber\\
&\geq \hspace{-1mm}\sum_{i}p_i\hspace{-1mm}\braket{\psi_i|A^2|\psi_i}\hspace{-1mm} -\frac{\left(\sum_{i}{p_i}|\hspace{-1mm}\braket{\phi|A|\psi_i}\hspace{-1mm}|^2\right)\hspace{-1mm}\big(\hspace{-1mm}\sum_{i}{p_i}\hspace{-1mm}\braket{\phi|\psi_i}\hspace{-1mm}\braket{\psi_i|\phi}\hspace{-1mm}\big)}{\braket{\phi|\rho|\phi}}\nonumber\\
&=\sum_{i}p_i\braket{\psi_i|A^2|\psi_i} -\sum_{i}{p_i}|\braket{\phi|A|\psi_i}|^2\nonumber\\
&=\sum_{i}p_i \left(\braket{\Delta A}^{\phi}_{\psi_i}\right)^2=\left(\braket{\Delta A}^{\phi}_{\rho}\right)^2{=\braket{Var A}_{\rho}^{\phi}},\nonumber
\end{align}
where we have used the Cauchy-Schwarz inequality for the complex numbers in the first inequality  and  Eq. (\ref{MixedPPS-2}) in the last line.  When $\rho$ is pure, equality  holds automatically.
\end{proof}
 {As $\braket{Var A}_{\rho}^{\phi}$   does neither increase nor decrease under  classical mixing of  quantum states, the inequality $\braket{Var A_w}_{\rho}^{\phi}\geq \braket{Var A}_{\rho}^{\phi}$  clearly implies  that under  classical mixing  of quantum states, the generalized variance $\braket{Var A_w}_{\rho}^{\phi}$  is always  non-decreasing. In fact, one can easily verify that $\braket{Var A_w}_{\rho}^{\phi}$ is sum of the \emph{quantum uncertainty} $\braket{Var A}_{\rho}^{\phi}$ and   the \emph{classical uncertainty} $C(\rho, A, \phi):=\braket{\phi|A\rho A|\phi}-|\braket{A_w}_{\rho}^{\phi}|^2\braket{\phi|\rho|\phi}$, both of which will be discussed in detail  in Sec. \ref{IV C}.}\par
It is important to note that,  in general  the equality in Eq. (\ref{MixedPPS-4}) does not imply that the pre-selection $\rho$ is pure.  In Sec. \ref{IV A} (see below), we  show that only  in the  qubit  system,  equality of Eq. (\ref{MixedPPS-4}) implies that  the pre-selection is a pure state. To make an equality in Eq. (\ref{MixedPPS-4})} in higher dimensional systems, we  need to put conditions on the observable and post-selection (see below in Sec. \ref{IV A} ).\par
The uncertainty relation (\ref{URPPS-1}) or (\ref{URPPS-5}) can be generalized   for mixed pre-selection $\rho$ also which is given by 
\begin{align}
 \left(\braket{\Delta{A}}^{\phi}_{\rho}\right)^2 \left(\braket{\Delta{B}}^{\phi}_{\rho}\right)^2 &\geq  \left[\frac{1}{2i}\braket{[A,B]}_{\rho}-\mathcal{I}m W_{AB}\right]^2,\label{MixedPPS-5}
\end{align}
where $W_{AB}=\braket{\phi|B\rho A|\phi}$. See the derivation of Eq. (\ref{MixedPPS-5}) in Appendix \ref{B}.  Eq. (\ref{MixedPPS-5}) holds also when the definition of standard deviation defined in Eq. (\ref{MixedPPS-3}) is considered
due  the \emph{Proposition 3}.\par
See TABLE \ref{Table} for the comparison  of different properties between standard quantum systems and PPS systems.

\section{Applications}\label{IV}
Suitably post-selected systems can offer some essential information regarding quantum systems. Below, we provide a few applications of standard deviations and uncertainty relations in PPS systems.

\subsection{Detection of mixedness of an unknown state}\label{IV A}
{Practically, partial information about a quantum state is often of great help. For example,  whether an interaction has taken place with the environment, one must verify the purity of the system's state.  Quantum state tomography (QST) is   the most resource intensive way  to verify the purity of a quantum state but here we provide   some results  that can be used to  detect purity of the quantum state  using   less resources compared to the QST. }\par
 We will use the inequality (\ref{MixedPPS-4}) in \emph{Proposition 3} to detect the mixedness of an unknown pre-selected state in a PPS system.
 The proofs of the following \emph{Lemmas} are given in  Appendix \ref{C}. 
 \begin{lemma}
     Qubit system: In the case of a two-level quantum system (i.e., a qubit), equality in Eq. (\ref{MixedPPS-4}) holds if and only if the pre-selected state $\rho$ is pure irrespective of choice of the observable $A$ and the post-selected state $\ket{\phi}$.
 \end{lemma}

 \begin{lemma}
     Qutrit system: If for an observable $A$ and a complete orthonormal basis $\{\ket{\phi_k}\}_{k=1}^3$ (to be used as post- selected states) of any three-level quantum system (i.e., a qutrit), and the condition $\braket{\phi_1|A|\phi_{2}} = 0$ also holds good, 
     then equality in Eq. (\ref{MixedPPS-4}) holds good if and only if the pre-selected state $\rho$ is pure.
 \end{lemma}

\begin{lemma}
    Qubit-qubit system: Consider any  two non orthogonal post-selections  $\ket{\phi_B}$ and $\ket{\phi^{\prime}_B}$ in the subsystem B.  For any observable $A$,  equality of $\braket{\Delta {(\hspace{-.5mm}A\hspace{-.5mm}\otimes \hspace{-.5mm}I)_w}}^{\phi\hspace{-.5mm}_{A\hspace{-.5mm}B}}_{\rho}$ and $\braket{\Delta {(\hspace{-.5mm}A\hspace{-.5mm}\otimes \hspace{-.5mm}I)}}^{\phi\hspace{-.5mm}_{A\hspace{-.5mm}B}}_{\rho}$ and separately of  $\braket{\Delta {(\hspace{-.5mm}A\hspace{-.5mm}\otimes \hspace{-.5mm}I)_w}}^{\phi^{\prime}\hspace{-.5mm}_{A\hspace{-.5mm}B}}_{\rho}$ and $\braket{\Delta {(\hspace{-.5mm}A\hspace{-.5mm}\otimes \hspace{-.5mm}I)}}^{\phi^{\prime}\hspace{-.5mm}_{A\hspace{-.5mm}B}}_{\rho}$  hold only when the $2\otimes 2$ pre-selected  state  $\rho$ is pure, where $\ket{\phi\hspace{-.5mm}_{A\hspace{-.5mm}B}}=\ket{\phi_A}\ket{\phi_B}$ and $\ket{\phi^{\prime}\hspace{-.5mm}_{A\hspace{-.5mm}B}}=\ket{\phi_A}\ket{\phi_B^{\prime}}$.   {Two non orthogonal post-selections  $\ket{\phi_B}$ and $\ket{\phi^{\prime}_B}$ in the subsystem B are required here due to the fact that there  exists an unique $2\otimes 2$ mixed  density matrix which  satisfies the equality of Eq. (\ref{MixedPPS-4}) }
\end{lemma}

    \begin{lemma}
Qubit-qutrit system: If for an observable $A$ and any complete orthonormal basis $\{\ket{\phi^k_A}\}_{k=1}^3$ (to be used as post-selected states) for a qutrit, and the condition $\braket{\phi_A^1|A|\phi_A^2} = 0$ is considered, 
then equality of $\braket{\Delta {(\hspace{-.5mm}A\hspace{-.5mm}\otimes \hspace{-.5mm}I)_w}}^{\phi\hspace{-.5mm}_{A\hspace{-.5mm}B}}_{\rho}$ and $\braket{\Delta {(\hspace{-.5mm}A\hspace{-.5mm}\otimes \hspace{-.5mm}I)}}^{\phi\hspace{-.5mm}_{A\hspace{-.5mm}B}}_{\rho}$ and separately of $\braket{\Delta {(\hspace{-.5mm}A\hspace{-.5mm}\otimes \hspace{-.5mm}I)_w}}^{\phi^{\prime}\hspace{-.5mm}_{A\hspace{-.5mm}B}}_{\rho}$ and $\braket{\Delta {(\hspace{-.5mm}A\hspace{-.5mm}\otimes \hspace{-.5mm}I)}}^{\phi^{\prime}\hspace{-.5mm}_{A\hspace{-.5mm}B}}_{\rho}$  hold  if and only if the $3 \otimes 2$ pre-selected  state $\rho$ is pure.
\end{lemma}

 Extension of  this method for higher dimensional systems  will require  more conditions to be imposed on the observable and post-selections. So it might be difficult to apply our method for higher dimensions. To overcome this difficulties,  Eq. (\ref{URPPS-1}) or (\ref{URPPS-5}) can be used to detect the mixedness of the initially prepared states. Note that,  Mal \emph{et al.} have used   the stronger version of the RHUR (\ref{RHUR-5}) to do so \cite{Mal-2013}. \\

\subsection{Stronger Uncertainty Relation}\label{IV B}
\emph{Motivation.} If, for example, the initially prepared state of the system is an eigenstate of one of the two incompatible observables $A$ and $B$, both the sides of the RHUR (\ref{RHUR-4}) becomes trivial (i.e., zero).  For certain states, a trivial  lower bound is always possible because the right side of the relation (\ref{RHUR-4}) contains the average of the commutator of incompatible observables. For such cases, the RHUR (\ref{RHUR-4})  does not capture the incompatibility of the non-commuting observables. One can think of adding  \emph{Schr$\ddot{o}$dinger's term}  in the RHUR  but still this can be become  trivial (e.g.,  when the prepared states is an eigenstate of either $A$ or $B$). So, none of them are unquestionably appropriate to capture the incompatibility of the non-commuting observables. \par

\begin{figure}[H]
\centering
\includegraphics[scale=0.2]{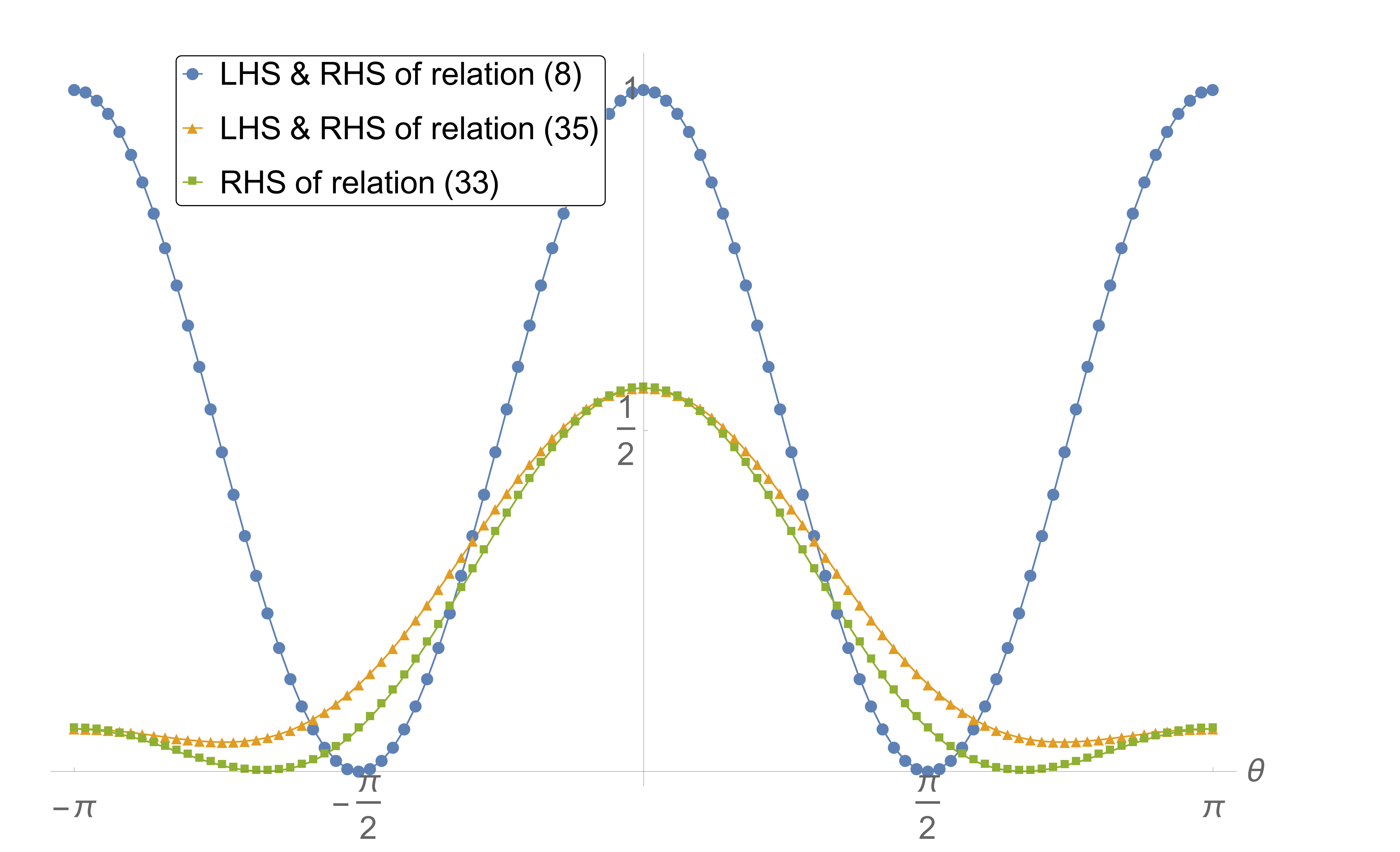}
\caption{Comparison between the  RHURs (\ref{RHUR-4}) and (\ref{RHUR-5}),  and the uncertainty relations (\ref{Strong-1}) and (\ref{Strong-3}).  We choose $A=\sigma_x$, $B=\sigma_y$ for a spin-1/2 particle and   $\ket{\psi}=\cos(\theta/2)\ket{0}+e^{i\xi}\sin(\theta/2)\ket{1}$,  $\ket{\phi}=\cos(\omega/2)\ket{0}+e^{i\eta}\sin(\omega/2)\ket{1}$  with $\xi=0$, $\omega=\pi/3$ and $\eta=\pi/5$. The blue curve is the LHS of the RHUR and for this particular case, it coincides with its lower bound i.e., RHS of RHUR. The orange curve is the LHS of  Eq. (\ref{Strong-3})  and for this particular case, it   coincides with the RHS of  Eq. (\ref{Strong-3}). The green curve is the RHS of Eq. (\ref{Strong-1}). Now, notice that, for $\theta=-\pi/2$ and $\pi/2$, the RHUR  becomes trivial while for the same values of $\theta$, the relation (\ref{Strong-1}) as well as the relation (\ref{Strong-3}) are nontrivial. For this particular choice of post-selection, both the relations (\ref{Strong-1}) and (\ref{Strong-3}) are stronger than the RHUR (\ref{RHUR-4}). Note that, the relation (\ref{Strong-3}) is the strongest under this condition as it is non-trivial for all the values of $\theta$. If, for the fixed values of $\theta$ and $\xi$,   the relations  (\ref{Strong-1}) and (\ref{Strong-3})   is trivial, then one should keep changing the values of $\omega$ and $\eta$ (i.e., by choosing the post-selection suitably) until they  become  nontrivial which is our main goal.}
\label{FIG1}
\end{figure}

It is Maccone and Pati \cite{Maccone-Pati} who considered a different uncertainty relation, based on the sum of the variances  $\braket{\Delta{A}}^{2}_{\psi} + \braket{\Delta{B}}^{2}_{\psi}$, that is guaranteed to be nontrivial  {(i.e., having  non-zero lower bound)} whenever the observables are incompatible on the given state $\ket{\psi}$.  But there are  shortcomings in the Maccone-Pati Uncertainty Relations (MPUR). {It is easy to show  that in two dimensional Hilbert space \cite{example-1}, if, for example, the initial state $\ket{\psi}$ of the system is an eigenstate of the observable $A$, then one finds that  the first inequality  $\braket{\Delta{A}}^{2}_{\psi} + \braket{\Delta{B}}^{2}_{\psi} \geq \pm i\braket{\psi|[A,B]|\psi} + \left|\braket{\psi|(A\pm iB)|\psi^{\perp}}\right|^2$  in MPUR becomes $\braket{\Delta{B}}^{2}_{\psi}\geq \braket{\Delta{B}}^{2}_{\psi}$, where $\ket{\psi^{\perp}}$ is arbitrary state orthogonal to $\ket{\psi}$. Similarly, it can be shown that the second inequality  $\braket{\Delta{A}}^{2}_{\psi} + \braket{\Delta{B}}^{2}_{\psi} \geq \frac{1}{2}\left|\braket{\psi^{\perp}_{A+B}|(A+B)|\psi}\right|^2$ in MPUR becomes $\braket{\Delta{B}}^{2}_{\psi}\geq \frac{1}{2}\braket{\Delta{B}}^{2}_{\psi}$, where $\ket{\psi^{\perp}_{A+B}}=(1/\braket{\Delta{(A+B)}}_{\psi})(A+B-\braket{A+B}_{\psi})\ket{\psi}$ and $\braket{\Delta{(A+B)}}_{\psi}^2=\braket{(A+B)^2}_{\psi}-\braket{A+B}_{\psi}^2$  for arbitrary dimensional Hilbert space if the initial state of the system is an eigenstate of the observable $A$ \cite{example-2}. It indicates that the first and second inequalities in MPUR for two and arbitrary dimensions, respectively, contain no information about the observable $A$ and are therefore of no practical significance.  In other words, we learn nothing new about the quantum system other than the trivial  fact that $\braket{\Delta{B}}_{\psi}$ is non-negative.  In addition,  if the initially prepared state $\ket{\psi}$ is unknown, then $\ket{\psi^{\perp}}$ is likewise unknown in the MPUR inequalities  and, so is the lower bound of MPUR.  The first  inequality in MPUR may be useful in a quantum system  with Hilbert spaces of more than two dimensions.  \par 
Here, we demonstrate that relations (\ref{URPPS-1}) and (\ref{URPPS-5})  can be used to solve the triviality problem of the RHUR and the problem with MPUR that we have mentioned above, i.e., these  uncertainty relations can provide non-trivial  information about the observable $A$. Even if the initially prepared state (pre-selection) $\ket{\psi}$ is unknown, the lower bound of our stronger uncertainty relation can be calculated.\par

Consider the relation (\ref{URPPS-1}) which,  using Eq. (\ref{SDPPS-5}), becomes
\begin{align}
\left(\hspace{-1mm}\braket{\Delta{A}}^{2}_{\psi}\hspace{-1mm} + \epsilon_A\hspace{-1mm} \right) \hspace{-1mm}\left(\hspace{-1mm}\braket{\Delta{B}}^{2}_{\psi} \hspace{-1mm}+ \epsilon_B\hspace{-1mm} \right)\hspace{-1mm} \geq \hspace{-1.5mm}\left[ \frac{1}{2i}\hspace{-1mm}\bra{\psi}\hspace{-1mm}[{A},{B}]\hspace{-1mm}\ket{\psi} \hspace{-1mm} - \mathcal{I}m\hspace{-1mm} \left(W_{AB}\right)\hspace{-0.5mm} \right]^2\hspace{-1mm},\label{Strong-1}
 \end{align} 
where $\epsilon_{X}={\braket{X}^2_{\psi}}  -|{\braket{X_w}^{\phi}_{\psi}}|^2|{\braket{\phi|\psi}}|^2$, with  $X=A$ or $B$. Now suppose $\ket{\psi}$ is an eigenstate  of $A$  then, the Eq. (\ref{Strong-1})  is nontrivial unless $\ket{\phi}=\ket{\psi}$, as, in the case when  $\ket{\phi}\ne\ket{\psi}$, the inequality (\ref{Strong-1}) becomes
\begin{align}
  \epsilon_A\left(\braket{\Delta{B}}^{2}_{\psi} + \epsilon_B \right) \geq & \left[ \mathcal{I}m \left(W_{AB}\right) \right]^2. \label{Strong-2}
 \end{align}\par

  Notice that, in the both sides of Eq. (\ref{Strong-2}), there is a  quantum state $\ket{\phi}$ which can be chosen  independently in the standard quantum system. So, it is always possible to choose a suitable $\ket{\phi}$ such that the relation (\ref{Strong-1}) is  nontrivial.  With   ``Schr$\ddot{o}$dinger's term", the relations (\ref{Strong-1}) and (\ref{Strong-2}) becomes 
\begin{align}
\left(\hspace{-1mm}\braket{\Delta{A}}^{2}_{\psi}\hspace{-1mm} + \epsilon_A\hspace{-1mm} \right) \hspace{-1mm}\left(\hspace{-1mm}\braket{\Delta{B}}^{2}_{\psi} \hspace{-1mm}+ \epsilon_B\hspace{-1mm} \right)\hspace{-1mm} &\geq\hspace{-1mm}\left[ \frac{1}{2i}\hspace{-1mm}\bra{\psi}\hspace{-1mm}[{A},{B}]\hspace{-1mm}\ket{\psi} \hspace{-1mm} - \mathcal{I}m\hspace{-1mm} \left(W_{AB}\right)\hspace{-0.5mm} \right]^2\nonumber\\
 &+ \left[\frac{1}{2}\hspace{-1mm}\braket{\psi|\{A,B\}|\psi}\hspace{-1mm} -\hspace{-1mm} \mathcal{R}e\hspace{-1mm}\left(W_{AB}\right)\hspace{-0.5mm}\right]^2\hspace{-1mm},\label{Strong-3} \\
  \epsilon_A\left(\hspace{-1mm}\braket{\Delta{B}}^{2}_{\psi} \hspace{-1mm}+ \epsilon_B\hspace{-1mm} \right)\hspace{-0.5mm} &\geq  \left[\mathcal{I}m \left(W_{AB}\right) \right]^2 \nonumber\\
 & + \left[\frac{1}{2}\hspace{-1mm}\braket{\psi|\{A,B\}|\psi}\hspace{-1mm} -\hspace{-1mm} \mathcal{R}e\hspace{-1mm}\left(W_{AB}\right)\hspace{-0.5mm}\right]^2\hspace{-1mm}, \label{Strong-4}
 \end{align}
respectively. As $\epsilon_A$ and $\epsilon_B$ can also be negative, the left hand side of relation (\ref{Strong-1}) can become lower  than the left hand side of relation (\ref{RHUR-4}). The same holds true for the right-hand side as well.  So, for a fixed $\ket{\psi}$, we always want to have  a nontrivial lower bound  from the relations (\ref{RHUR-4}) and  (\ref{Strong-1}) which can be combined in  a single uncertainty relation i.e.,  the \emph{stronger uncertainty relation} 
\begin{align}
max\left\{ \mathcal{L}_{RH}, \mathcal{L}_{PPS} \right\} \geq max\left\{ \mathcal{R}_{RH}, \mathcal{R}_{PPS} \right\}, \nonumber
 \end{align} 
where $\mathcal{L}_{RH}=\braket{\Delta{A}}_{\psi}^{{2}}\hspace{-1mm}\braket{\Delta{B}}_{\psi}^{{2}}$, $\mathcal{L}_{PPS}=\left(\hspace{-1mm}\braket{\Delta{A}}^{2}_{\psi}\hspace{-1mm} + \epsilon_A\hspace{-1mm} \right) \hspace{-1mm}\left(\hspace{-1mm}\braket{\Delta{B}}^{2}_{\psi} \hspace{-1mm}+ \epsilon_B\hspace{-1mm} \right)\hspace{-1mm}$ , $\mathcal{R}_{RH}=\left[\frac{1}{2i}\braket{\psi|[A,B]|\psi}\right]^2$ and $\mathcal{R}_{PPS}=\left[ \frac{1}{2i}\bra{\psi}[{A},{B}]\ket{\psi}  - \mathcal{I}m \left(W_{AB}\right) \right]^2$.

In  FIG. \ref{FIG1},  comparison between the relations (\ref{RHUR-4}) and, both (\ref{Strong-1}) and (\ref{Strong-3})  is shown.  Eqs. (\ref{Strong-2}) and (\ref{Strong-4}) capture the informations about the operator $A$ when the initially prepared state $\ket{\psi}$ is one of the eigenstates of $A$, while MPUR fails to capture  such informations which we have already discussed.\par
Moreover, even if the initial state (i.e., pre-selection) is unknown, the lower bound  of the uncertainty relation (\ref{Strong-1})  can be calculated experimentally and in that case  we need the average value of the hermitian  operator $\frac{1}{i}[A,B]$ and weak values of the operators $A$ and $B$. \par
Sum uncertainty relation in the PPS system  can also be used to obtain stronger uncertainty relation in the standard quantum system. One can easily show that  
\begin{align}
\left(\hspace{-1mm}\braket{\Delta{A}}^{2}_{\psi}\hspace{-1mm} + \epsilon_A\hspace{-1mm} \right) \hspace{-1mm}+\left(\hspace{-1mm}\braket{\Delta{B}}^{2}_{\psi} \hspace{-1mm}+ \epsilon_B\hspace{-1mm} \right)\geq \pm\Big(i\braket{\psi|[A,B]|\psi}\nonumber\\
-2 \mathcal{I}m(W_{AB})\Big)\nonumber
\end{align}
holds in Eq. (\ref{UEPPS-4}) in  Theorem 3 by discarding the summation part which is always a positive number. This inequality  remains strong against when $\ket{\psi}$ is one of the eigenstates of $A$ by suitably choosing post-selection $\ket{\phi}$. \par 
 \subsection{Purely quantum uncertainty relation}\label{IV C}
\emph{Motivation.} In practice, it is not always possible to carry out quantum mechanical tasks with pure states because of interactions with the environment. Because the mixed initial prepared state is a classical mixture of pure quantum states, any task or measurement involves a hybrid of classical and quantum parts. In modern technologies, it is considered that   quantum advantage is more effective and superior to classical advantage. Hence, a hybrid of a  quantum and classical component may be less advantageous than a  quantum component alone. For example, the uncertainty of an observable $A$ in standard quantum system increases in general under classical mixing of quantum states i.e.,  $\braket{\Delta A}^2_{\rho}\geq \sum_i p_i \braket{\Delta A}^2_{\psi_i}$  (where $\rho=\sum_{i}p_i\ket{\psi_i}\bra{\psi_i}$) and this is disadvantageous in the sense that the knowledge about the  observable is more uncertain than only when the average of the  pure state uncertainties are considered.  The uncertainty that one gets to see due to (classical) mixing of pure states is considered here as `classical uncertainty'.   The standard deviation  $\braket{\Delta A}_{\rho}$ can be  referred as the  hybrid of classical and quantum uncertainties and hence the RHUR (\ref{RHUR-4}) can be considered as the hybrid uncertainty relation in the standard quantum systems. \par
Purely quantum uncertainty relation, a crucial component of the quantum world, may be very useful, but in order to obtain it, the classical uncertainty  must be eliminated from  the hybrid uncertainty relation. To do this, we first need to determine the purely quantum uncertainty of an observable, which can be done in a number of ways, such as by eliminating the classical component of the hybrid uncertainty or by specifying a purely quantum uncertainty straight away.\par
 Any measure of  purely quantum uncertainty should have at least the following  intuitive and expected property (below `$\Phi$' represents  sometimes quantum observables, sometimes states, etc for  different  type of quantum mechanical systems. For example, if the system is a PPS system, then `$\Phi$'  is the post-selection $\ket{\phi}$. If the system is a standard quantum system, then the term `$\Phi$' disappears):\\
($i$) \emph{ Quantum uncertainty should not be affected (neither increasing nor decreasing) by the classical mixing of quantum states i.e.,}
\begin{align}
\mathcal{Q}(\rho, A, \Phi) \hspace{-1mm}= \hspace{-1mm}\sum_i p_i \mathcal{Q}(\psi_i,A,\Phi), \hspace{1mm} where \hspace{1mm} \rho= \hspace{-1mm}\sum_{i}p_i\ket{\psi_i}\bra{\psi_i}.\nonumber
\end{align}
Here $\mathcal{Q}(\rho, A,\Phi)$ is some measure of  purely quantum uncertainty of the observable $A$ for a given $\rho$. There might exist some other properties depending upon the nature of the system (e.g., standard systems, PPS systems, etc) but we emphasize  that  the most important property of a   purely quantum uncertainty should be ($i$).
\par
{It is seen  that  the variance of $A$ in PPS system i.e.,   $\braket{Var A}_{\rho}^{\phi}=(\braket{\Delta A}^{\phi}_{\rho})^2$ is a  purely quantum uncertainty which satisfies the property ($i$).  Now the purely   quantum mechanical uncertainty relation in this regard is Eq. (\ref{MixedPPS-5}}).\par
{As can be seen from Eq. (\ref{MixedPPS-4}}) for mixed states, the second definition of the  variance of $A$ i.e., $\braket{Var A_w}_{\rho}^{\phi}=(\braket{\Delta A_w}^{\phi}_{\rho})^2$ in the PPS system defined in  Eq. (\ref{MixedPPS-3}) is a hybrid uncertainty. Hence, the  uncertainty in PPS system   based on weak value has both classical and quantum parts. When measurement is carried out in the PPS system and weak values are involved, classical uncertainty may be crucial in determining how much classicality (in the form of classical uncertainty) the mixed state $\rho$ possesses.  Mixed states with less classicality  should have more quantumness (in the form of quantum uncertainty), and vice versa. To distinguish classical uncertainty from the hybrid uncertainty $\braket{Var A_w}_{\rho}^{\phi}$, we subtract the   quantum uncertainty $\braket{Var A}_{\rho}^{\phi}$ from it i.e.,
\begin{align}
C(\rho, A, \phi)=\left(\braket{\Delta A_w}^{\phi}_{\rho}\right)^2 - \left(\braket{\Delta A}^{\phi}_{\rho}\right)^2.\label{PureUR-1}
\end{align}
This is one of the good measures of classical uncertainty which should have some intuitive and expected properties:\\
($i$) $\mathcal{C}(\rho, A, \Phi)\geq 0$ for a quantum state  $\rho$,\\
($ii$) $\mathcal{C}(\rho, A, \Phi)=0$ when $\rho=\ket{\psi}\bra{\psi}$ (absence of classical mixing),\\
($iii$) Total classical uncertainty of disjoint systems should be the sum of individual systems's classical uncertainties:
\begin{align}
\mathcal{C}(\rho,A_1\hspace{-1mm}\otimes \hspace{-1mm}I+I\hspace{-1mm}\otimes \hspace{-1mm}A_2,\Phi)=\mathcal{C}(\rho,A_1\hspace{-1mm}\otimes \hspace{-1mm}I,\Phi)+\mathcal{C}(\rho,I\hspace{-1mm}\otimes \hspace{-1mm}A_2,\Phi),\nonumber
\end{align}
when $\rho=\rho_1\otimes\rho_2$.
}\par 
One can show that all the properties ($i$)-($iii$) of  classical uncertainty are satisfied by $C(\rho, A, \phi)$ defined in Eq. (\ref{PureUR-1}). Particularly, property ($iii$) is satisfied by taking $\Phi=\ket{\phi_1}\ket{\phi_2}$. Here, $\ket{\phi_1}$ and $\ket{\phi_2}$ are post-selections of the two disjoint systems, respectively.\par
{There are some works by Luo and  other authors regarding the purely  quantum uncertainty relation. Initial attempt was made by Luo  and Zhang \cite{Luo-Zhang-2004} to obtain  uncertainty relation by using skew information (introduced by Wigner and Yanase \cite{Wigner-Yanase-1993}) but it was found to be incorrect in general \cite{Kuriyama-2005}.  Later, another attempt was made by Luo himself \cite{Luo-2005}, which is obtained by discarding the classical part from the hybrid  uncertainty relation using skew information. {But this uncertainty relation can’t be guaranteed to be an intrinsically quantum uncertainty relation (according to property $(i)$) as the uncertainty they claim to be a quantum uncertainty is a product of skew information (which is a convex function under the mixing of quantum states) and a concave function under the same mixing}.  After that a series of successful and failed attempts  was performed by modifying the works of Luo and    other authors \cite{Li-2009,Furuichi-2009,Yanagi-2010,Kenjiro_Yanagi-2010}. }\par
Instead, we have given a  quantum uncertainty relation  although it is based on pre- and post-selections which is different from the standard quantum mechanics but a  quantumness can be seen in the relation (\ref{MixedPPS-5}).
\subsection{ State dependent  tighter uncertainty relations in standard systems}
 The RHUR (\ref{RHUR-4}) or (\ref{RHUR-5}) is known  not to be the  tight one. Some existing tighter bounds are given in \cite{Yao-2015,Maccone-Pati,Qiao-2016}. The drawback of these tighter uncertainty relations is that their lower bounds   depend on the states perpendicular to the given state of the system. If the given state is unknown, then the  lower bound of these uncertainty relations also remain  unknown.\par
 Here we show  that by the use  of arbitrary  post-selected state $\ket{\phi}$, the lower bound of the RHUR based on sum uncertainties  can be made arbitrarily tight and even if the given state (i.e., pre-selection here) is unknown, the lower bound of our tighter uncertainty relation can be obtained in experiments.
 \begin{theorem}
Let $\rho\in \mathcal{L}(\mathcal{H})$ be the density operator of the standard quantum system, then  the sum of the standard deviations of  two non-commuting observables $A$, $B\in \mathcal{L}(\mathcal{H})$ satisfies  
\begin{align}
\hspace{-2mm}\braket{\Delta{A}}_{\rho}^{2}+\braket{\Delta{B}}_{\rho}^{2}\geq\pm \,i\,Tr([A,B]\rho)+
 \braket{\phi|C^{\dagger}_{\pm}\rho C_{\pm}|\phi},\label{Tight-1}
\end{align}
where $C_{\pm}=A\pm iB-\braket{A\pm iB}_{\rho}I$ and  the `$\pm$' sign is taken in such a way that the first term in the right hand side  is always positive.
\end{theorem}
\begin{proof}
Considering   Eq. (\ref{UEPPS-4}) for pre-selection $\ket{\psi_j}$ and multiply by `$p_j$', and  then after  summing  over `$j$',  we have  
\begin{align}
&\sum_{j}p_j\left(\braket{\Delta{A}}^{\phi}_{\psi_j}\right)^2+\sum_{j}p_j\left(\braket{\Delta{B}}^{\phi}_{\psi_j}\right)^2\nonumber\\
&=\pm i \hspace{-1mm}\sum_jp_j\hspace{-1mm}\braket{\psi_j|[A,B]|\psi_j}\mp2\mathcal{I}m\Big(\sum_jp_j\hspace{-1mm}\braket{\phi|B|\psi_j}\hspace{-1mm}\braket{\psi_j|A|\phi}\hspace{-1.5mm}\Big)\nonumber\\
&\hspace{4mm}+\hspace{-1mm}\sum_{k=1}^{d-1}\sum_jp_j\hspace{-1mm}\braket{\phi_k^{\perp}|(A\pm iB)|\psi_j}\hspace{-1mm}\braket{\psi_j|(A\mp iB)|\phi_k^{\perp}},\label{Tight-2}
\end{align}
where we have used $W_{AB}=\braket{\phi|B|\psi_j}\hspace{-1mm}\braket{\psi_j|A|\phi}$. By using Eq. (\ref{MixedPPS-2}) for $A$ and $B$ when $\rho=\sum_jp_j\ketbra{\psi_j}{\psi_j}$, we have 
\begin{align}
&\left(\braket{\Delta{A}}^{\phi}_{\rho}\right)^2+\left(\braket{\Delta{B}}^{\phi}_{\rho}\right)^2\nonumber\\
&=\pm i \Tr([A,B]\rho)\mp2\mathcal{I}m\Big(\hspace{-1mm}\braket{\phi|B\rho A|\phi}\hspace{-1mm}\Big)\nonumber\\
&\hspace{2.5cm}+\sum_{k=1}^{d-1}\braket{\phi_k^{\perp}|(A\pm iB)\rho (A\mp iB)|\phi_k^{\perp}}\nonumber\\
&=\pm i \Tr([A,B]\rho)\hspace{-1mm}+\hspace{-1mm}\braket{\phi|(A\pm iB)\rho(A\mp iB)|\phi}\hspace{-1mm}-\hspace{-1mm}\braket{\phi|A\rho A|\phi}\nonumber\\
&\hspace{1mm}-\braket{\phi|B\rho B|\phi}+\sum_{k=1}^{d-1}\braket{\phi_k^{\perp}|(A\pm iB)\rho (A\mp iB)|\phi_k^{\perp}},\label{Tight-3}
\end{align}
where  $\mp2\mathcal{I}m(\braket{\phi|B\rho A|\phi})=\pm i(\braket{\phi|B\rho A|\phi}-\braket{\phi|A\rho B|\phi})$ = $\braket{\phi|(A\pm iB)\rho(A\mp iB)|\phi}-\braket{\phi|A\rho A|\phi}-\braket{\phi|B\rho B|\phi}$ has been used. Now put $\left(\braket{\Delta{A}}^{\phi}_{\rho}\right)^2=\Tr(A^2\rho)-\braket{\phi|A\rho A|\phi}$ defined in Eq. (\ref{MixedPPS-1}) (similarly for $B$ also) and after subtracting $\Tr(A\rho)^2+\Tr(B\rho)^2$ from both sides of Eq. (\ref{Tight-3}) and using $\ketbra{\phi}{\phi}+\sum_{k=1}^{d-1}\ketbra{\phi_k^{\perp}}{\phi^{\perp}_k}=I$, we have
\begin{align}
&\braket{\Delta A}_{\rho}^2+\braket{\Delta B}_{\rho}^2\nonumber\\
&=\pm i \Tr([A,B]\rho)+\Tr[(A\pm iB)(A\mp iB)\rho]\nonumber\\
&\hspace{5.3cm}-\Tr(A\rho)^2-\Tr(B\rho)^2\nonumber\\
&=\pm i\Tr([A,B]\rho)+\Tr[(A\pm iB)(A\mp iB)\rho]\nonumber\\
&\hspace{3.9cm}-\Tr[(A\pm iB)\rho]\Tr[(A\mp iB)\rho] \nonumber\\
&=\pm i\Tr([A,B]\rho)+\Tr(M^{\dagger}_{\mp}M_{\mp}\rho)-|\Tr(M_{\mp}\rho)|^2\nonumber\\
&=\pm i\Tr([A,B]\rho)+\Tr[(M_{\mp}-\hspace{-1mm}\braket{M_{\mp}}_{\rho}\hspace{-1mm}I)^{\dagger}(M_{\mp}-\hspace{-1mm}\braket{M_{\mp}}_{\rho}\hspace{-1mm}I)\rho],\label{Tight-4}
\end{align}
where $M_{\mp}=A\mp iB$.  Now let $C_{\pm}=M_{\pm}-\braket{M_{\pm}}_{\rho}I$ then Eq. (\ref{Tight-4}) can be rewritten as 
\begin{align}
\braket{\Delta A}_{\rho}^2+\braket{\Delta B}_{\rho}^2=&\pm i\Tr([A,B]\rho)
+\braket{\phi|C_{\pm}^{\dagger}\rho C_{\pm}|\phi}\nonumber\\
&\hspace{1.8cm}+\sum_{i}^{d-1}\braket{\phi^{\perp}_i|C_{\pm}^{\dagger}\rho C_{\pm}|\phi^{\perp}_i},\nonumber
\end{align}
where \{$\ket{\phi}$,  \{$\ket{\phi^{\perp}_i}\}^{d-1}_{i=1}\}$ is an orthonormal basis in $\mathcal{H}$. By discarding the summation  term which is always a positive number  in the above equation,  we obtain the inequality (\ref{Tight-1}). 
\end{proof}
Notice that, the lower bound of Eq. (\ref{Tight-1}) has different non-zero values depending on different choices of the post-selections $\ket{\phi}$. The inequality (\ref{Tight-1}) becomes an equality when $\ket{\phi}\propto (A\pm iB-\braket{A\pm iB}_{\rho}I)\ket{\psi}$, where $\rho=\ket{\psi}\bra{\psi}$ is a pure state. In the references \cite{Yao-2015,Maccone-Pati,Qiao-2016},  the lower bound of the sum uncertainty relation depends on  the state  orthogonal to the  initial pure state, and if  the  initial  state is a mixed state, then  the lower bound can not always be  computed at least for the full rank density matrix. The reason is that we can not find a state which is orthogonal to all the eigenstates of a full rank density matrix. Moreover,  if the initial density matrix is unknown then computing the lower bound will be hard. In contrast, Eq. (\ref{Tight-1}) doesn't have such issues as the first  and second terms in right hand side of Eq. (\ref{Tight-1}) are  the average values of the hermitian operators  $i[A,B]$ and $(A\pm iB-\braket{A\pm iB}_{\rho}I)\ket{\phi}\bra{\phi}(A\mp iB-\braket{A\mp iB}_{\rho}I)$  in the  state $\rho$, respectively, where  $\braket{A\pm iB}_{\rho}=\braket{A}_{\rho}\pm i\braket{B}_{\rho}$. All of them can be obtained in experiments even if $\rho$ is unknown.

    \subsection{Tighter upper bound for out-of-time-order correlators}
 Recently Bong \emph{et al.} \cite{Bong-2018} used the RHUR for unitary operators to give  upper  bound for out-of-time-order correlators (OTOC) which is defined by $F=Tr[(W_t^{\dagger}V^{\dagger}W_tV)\rho]$, where $V$ and $W_t$ are fixed and time dependent unitary operators, respectively. The OTOC   diagnoses the spread of quantum information by measuring how quickly two commuting operators $V$  and $W$ fail to commute, which is  quantified by $\braket{|[W_t,V]|^2}_{\rho}=2(1-Re[F])$, where $|X|^2=X^{\dagger}X$. The OTOC has strong connection with chaos and information scrambling \cite{Swingle-2018,Swingle-2016,Halpern-2019} and  also with high energy physics \cite{Susskind-2008,Hayden-2013,Stanford-2016,Hosur-2016}.  It is known that OTOC's upper bound is essential for limiting how quickly many-body entanglement can generate \cite{Susskind-2008,Hayden-2013,Stanford-2016}.  The standard upper bound  for modulus of the  OTOC given by Bong \emph{et al.} \cite{Bong-2018} is $|F|\leq cos(\theta_{VW_t}-\theta_{W_tV})$,  where $\theta_{VW_t}=cos^{-1}|Tr(\rho VW_t)|$,  $\theta_{W_tV}=cos^{-1}|Tr(\rho W_tV)|$.    \par
 Here, we show that uncertainty relation in PPS system  for unitary operators  can be used to derive  tighter upper bound for the  OTOC. 
 \begin{theorem}
     Let $\rho\in \mathcal{L}(\mathcal{H})$ be the system's state  and  $\ket{\phi}$ be  any  arbitrary state, then modulus of the  OTOC $F=Tr[(W_t^{\dagger}V^{\dagger}W_tV)\rho]$ for fixed and time dependent unitary operators $V$, $W_t\in \mathcal{L}(\mathcal{H})$, respectively  is upper bounded by 
 \begin{align}
 |F|=|\braket{W_t^{\dagger}V^{\dagger}W_tV}|\leq cos(\theta_{VW_t}^{\phi}-\theta_{W_tV}^{\phi}), \label{OTOC-1}
 \end{align}
 where $\theta_{VW_t}^{\phi}=cos^{-1}||\sqrt{\rho}(VW_t)^{\dagger}\ket{\phi}||$ and  $\theta_{W_tV}^{\phi}=cos^{-1}||\sqrt{\rho}(W_tV)^{\dagger}\ket{\phi}||$.   Here, $||.||$ defines a vector norm.
\end{theorem}
\begin{proof}
For a given  mixed  state $\rho$   and  arbitrary state $\ket{\phi}$  which we consider to be pre- and post-selections, respectively, the standard deviation $\braket{\Delta X}_{\rho}^{\phi}$  of any operator $X$  in the PPS system is defined  as $\left(\braket{\Delta X}_{\rho}^{\phi}\right)^2=Tr(XX^{\dagger}\rho)-\braket{\phi|X^{\dagger}\rho X|\phi}=Tr\left((\sqrt{\rho}X_0^{\phi})^{\dagger}\sqrt{\rho}X_0^{\phi}\right)=||\sqrt{\rho}X_0^{\phi}||^2_{F}$, where $X_0^{\phi}=X-X\ket{\phi}\bra{\phi}$ and $||A||_F=\sqrt{Tr(A^{\dagger}A)}$ denotes the Frobenius norm of the operator $A$. When  $X$ is a hermitian operator, $\braket{\Delta X}_{\rho}^{\phi}$ becomes the standard deviation of $X$ defined in Eq. (\ref{MixedPPS-1}). Now consider $X$ to be  unitary operators $U$ and $V$. So, we can derive uncertainty relation  for two unitary operators $U$ and $V$ using the Cauchy-Schwarz inequality for operators with Frobenius norm  when the system is in pre- and post-selections $\rho$ and $\ket{\phi}$, respectively as
\begin{align}
\braket{\Delta U}_{\rho}^{\phi} \braket{\Delta V}_{\rho}^{\phi} \geq \left|Tr\left[(\sqrt{\rho}U_0^{\phi})^{\dagger}\sqrt{\rho}V_0^{\phi}\right]\right|\nonumber\\
=\left| Tr(VU^{\dagger}\rho)-\braket{\phi|U^{\dagger}\rho V|\phi} \right|,\label{OTOC-2}
\end{align} 
where $\braket{\Delta U}_{\rho}^{\phi}=\sqrt{1-\braket{\phi|U^{\dagger}\rho U|\phi}}$ and similarly for $V$ also.  
Now, by  replacing $U\rightarrow V^{\dagger}W_t^{\dagger}$ and  $V\rightarrow W_t^{\dagger}V^{\dagger}$,  (\ref{OTOC-2}) becomes
\begin{align}
&|Tr(W_t^{\dagger}V^{\dagger}W_tV\rho)|\nonumber\\
&\leq |\braket{\phi|(V^{\dagger}W_t^{\dagger})^{\dagger}\rho W_t^{\dagger}V^{\dagger}|\phi}|+ \braket{\Delta (V^{\dagger}W_t^{\dagger})}_{\rho}^{\phi} \braket{\Delta (W_t^{\dagger}V^{\dagger})}_{\rho}^{\phi}\nonumber\\
&\leq ||\sqrt{\rho}V^{\dagger}W_t^{\dagger}\ket{\phi}|| ||\sqrt{\rho}W_t^{\dagger}V^{\dagger}\ket{\phi}||\nonumber\\
&\hspace{1cm}+ \sqrt{1-||\sqrt{\rho}V^{\dagger}W_t^{\dagger}\ket{\phi}||^2}\sqrt{1-||\sqrt{\rho}W_t^{\dagger}V^{\dagger}\ket{\phi}||^2},\label{OTOC-3}
\end{align}
where we used the Cauchy-Schwarz inequality for vectors and  $\braket{\Delta(V^{\dagger}W_t^{\dagger})}_{\rho}^{\phi}=\sqrt{1-||\sqrt{\rho}V^{\dagger}W_t^{\dagger}\ket{\phi}||^2}$  and $\braket{\Delta(W_t^{\dagger}V^{\dagger})}_{\rho}^{\phi}=\sqrt{1-||\sqrt{\rho}W_t^{\dagger}V^{\dagger}\ket{\phi}||^2}$, where $||\ket{\chi}||=\sqrt{\braket{\chi|\chi}}$ denotes vector norm.\par
Now, by setting  $||\sqrt{\rho}(VW_t)^{\dagger}\ket{\phi}||=cos\theta_{VW_t}^{\phi}$ and  $||\sqrt{\rho}(W_tV)^{\dagger}\ket{\phi}||=cos\theta_{W_tV}^{\phi}$ in (\ref{OTOC-3}), the inequality (\ref{OTOC-1}) is proved.    
\end{proof}
 \begin{figure}[H]
 \includegraphics[width=10.5cm]{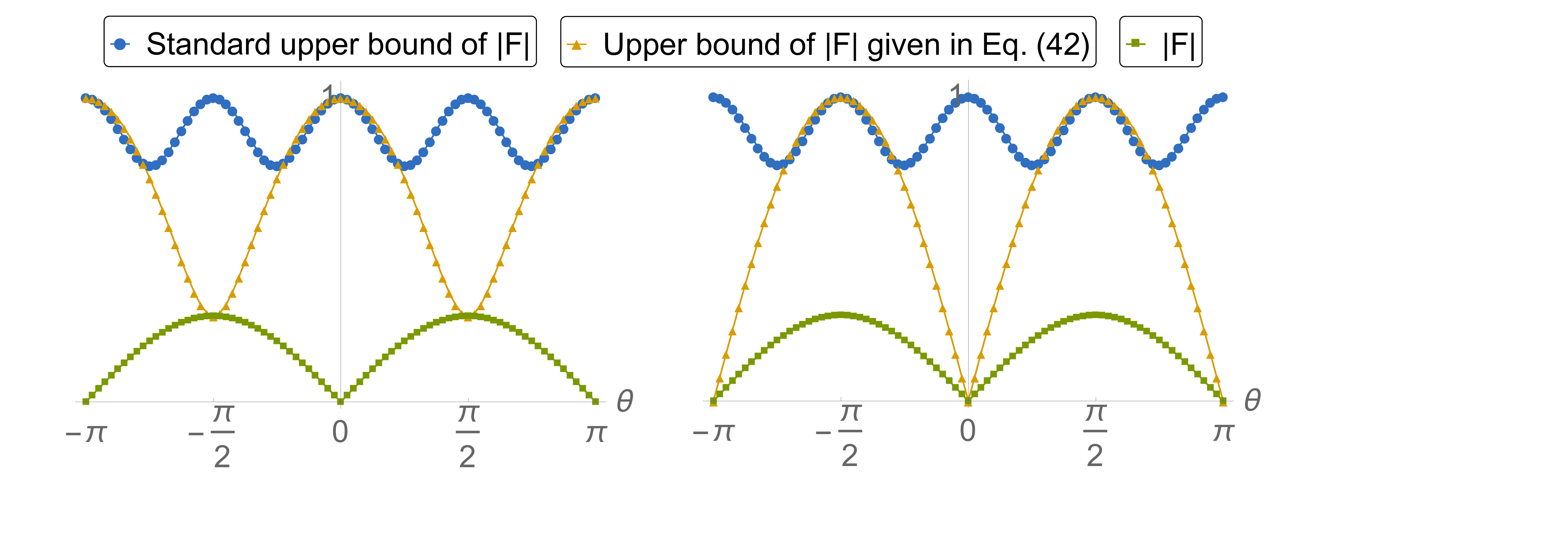}
 \centering
\caption{For both the figures, the blue  curve is the standard upper bound for $|F|$ given by Bong \emph{et al.} \cite{Bong-2018} and the green curve is $|F|$. We have considered $V=\sigma_z$ and $W_t=\frac{1}{\sqrt{2}}\big(\begin{smallmatrix}
  1 & 1\\
  -i & i
\end{smallmatrix}\big)$ for a fixed time.  Initially prepared state is $\ket{\psi}=cos(\theta/2)\ket{0}+e^{i\pi/11}sin(\theta/2)\ket{1}$. Now in the left figure, the orange curve is the upper bound of $|F|$ given in  Eq. (\ref{OTOC-1}) when the post-selection is $\ket{\phi_1}=cos(\pi/2)\ket{0}+e^{i\pi/2}sin(\pi/2)\ket{1}$. In the right  figure, the orange curve is the upper bound of $|F|$ given in the Eq. (\ref{OTOC-1}) when the post-selection is $\ket{\phi_2}=cos(\pi/4)\ket{0}+e^{i\pi/2}sin(\pi/4)\ket{1}$.  Here  for two (or more) different post-selections, it is   clearly seen  that the upper bound given in  Eq. (\ref{OTOC-1}) is tighter than  the standard upper bound given by Bong \emph{et al.} \cite{Bong-2018}.  } 
\label{FIG2}
\end{figure}
 In Fig. \ref{FIG2}, it is shown that by suitably choosing $\ket{\phi}$, the upper bound of $|F|$ in Eq. (\ref{OTOC-1}) can be made tighter than the standard upper bound given by Bong \emph{et al.} \cite{Bong-2018}. Hence, we conclude that the tighter upper bound for the modulus of the  OTOC is 
  \begin{align}
 |F|\hspace{-1mm}\leq min \left\{\underset{\phi}{min}\{\hspace{0mm}cos(\theta_{VW_t}^{\phi}-\theta_{W_tV}^{\phi})\hspace{0mm}\}\hspace{0mm}, cos(\theta_{VW_t}-\theta_{W_tV})\hspace{-1mm}\right\}\hspace{-1mm}.\nonumber
 \end{align}\par

  \section{Conclusion}\label{V}
  We have defined standard deviation of an observable in a PPS system, interpreted it geometrically as well as informationally  from  the perspective of  weak PPS measurements and subsequently derived the Robertson-Heisenberg like uncertainty relation for two non commuting observables. Such   uncertainty relations in  PPS system  impose limitations on the joint sharp preparation of pre- and post-selected states for two incompatible observables. We provided the necessary and sufficient condition for zero uncertainty of an observable and show its usefulness in achieving optimized Fisher information in quantum metrology. We have derived both product and sum uncertainty equalities from which a series of uncertainty inequalities can be obtained.   The generalization of uncertainty relation for mixed pre-selection  in PPS system has also been  discussed. We have demonstrated that the PPS system can exhibit more bizarre behaviors than the usual ones. For instances,  it is possible in PPS system that measurement  of two compatible observables can disturb each other's measurement results i.e., the lower bound in the uncertainty relation can be made non zero by suitably choosing post-selections. A similar  property  in PPS system  was first shown by  Vaidman \cite{Vaidman-1993}.  It is also possible that  a quantum state (pre-selection) can be prepared in a PPS system  for which both of the standard deviations of incompatible observables are zero although this is not possible  in a standard quantum system (see section \ref{III B}).
  \par
 The standard deviation and uncertainty relation in the PPS system have been used to provide physical applications. ($i$)  We have used two different definitions of the standard  deviations in the PPS system to detect purity of an unknown state. ($ii$) The uncertainty relation in the PPS system is used to derive the stronger uncertainty relation (i.e., nontrivial for all possible choices of initially prepared states) in the standard quantum system. For two dimensional quantum system,  the stronger uncertainty relation by Maccone-Pati \cite{Maccone-Pati} fails to provide the informations about the incompatible observables when the system state is an eigenstate of either observables. We have shown that our stronger uncertainty relation overcomes this shortcoming of Maccone-Pati uncertainty relation. ($iii$)  Since the variance in the PPS system remains unaffected (i.e., neither increases nor decreases) by the classical mixing of quantum states, we have concluded that the uncertainty relation in the PPS system is a purely quantum uncertainty relation. In contrast, variance in the standard system increases in general under the classical mixing of quantum states. Following this observation we have provided a measure of classical uncertainty whose less value implies more purely quantum uncertainty. ($iv$) Tighter sum uncertainty relation in the standard quantum system  has been  derived  where the tightness depends on the post-selection. ($v$)  Uncertainty relation in PPS system for two  unitary operators  has been used to provide   tighter upper bound for out-of-time-order correlators.
  \par
   {Future directions: ($i$) it will be interesting if the global minimum for sum of uncertainties of non-commuting observables in the PPS system exists because that can be used to detect  entanglement by suitably choosing post-selections, Similar to the work by Hofmann and Takeuchi \cite{Hofmann-2003}.  ($ii$) Applications and  implications of the ideas like  `zero uncertainty' and `joint sharp preparation of a quantum state for non-commuting observables' need more attention. ($iii$) This is a matter of further study if the uncertainty relation (\ref{URPPS-1}) in PPS system has applications similar to the RHUR (\ref{RHUR-4}), such as quantum metrology, spin squeezing, improving the accuracy of phase measurement in quantum interferometers, etc. ($iv$) We have derived  the condition for the ``intelligent pre- and post-selected states" to achieve the minimum bound of the  uncertainty relation in the PPS system and  intelligent pre- and post-selected states can be exploited to get highly precise phase measurements because  many theoretical and experimental efforts have been made in recent years involving the minimum uncertainty states (for which the RHUR saturates) and the spin-squeezing states in the standard quantum systems (see, for example, \cite{Pezze-2018,Ma-2011,Mlodinow-1993}) for precise phase measurements.\\
 {\bf{Acknowledgment}}: We would like to thank  Klaus M$ø$lmer and  David R. M. Arvidsson-Shukur  for bringing the references  \cite{Gammelmark-2013,Tan-2015,Bao-2020,Molmer-2021} and \cite{Shukur-2020}, respectively to our attention. {We  acknowledge Debmalya Das for valuable comments on this work.}

\section{APPENDICES}
\appendix

\section{}\label{A}
Here we derive the condition for which the inequality (\ref{URPPS-1}) saturates. In the Cauchy-Schwarz inequality Eq. (\ref{URPPS-2}), the remainder  and the real term  to be vanished for the equality condition of Eq. (\ref{URPPS-1}) i.e.,
 \begin{align}
  \ket{\widetilde{\phi}^{\perp}_{A\psi}} - \frac{{{\braket{\widetilde{\phi}^{\perp}_{B\psi}|\widetilde{\phi}^{\perp}_{A\psi}}}}}{{{\braket{\widetilde{\phi}^{\perp}_{B\psi}|\widetilde{\phi}^{\perp}_{B\psi}}}}}{{\ket{\widetilde{\phi}^{\perp}_{B\psi}}}}=0,\label{A1}\\
  {{\braket{\widetilde{\phi}^{\perp}_{A\psi}|\widetilde{\phi}^{\perp}_{B\psi}}}}+{{\braket{\widetilde{\phi}^{\perp}_{B\psi}|\widetilde{\phi}^{\perp}_{A\psi}}}}=0\label{A2}.
 \end{align}
Now take the inner product between   $\bra{\phi^{\perp}_{A\psi}}$ and Eq. (\ref{A1}), and use the condition (\ref{A2}), then we have
 \begin{align}
  \braket{\widetilde{\phi}^{\perp}_{A\psi}|\widetilde{\phi}^{\perp}_{A\psi}} + \frac{\left({{\braket{\widetilde{\phi}^{\perp}_{B\psi}|\widetilde{\phi}^{\perp}_{A\psi}}}}\right)^2}{{{\braket{\widetilde{\phi}^{\perp}_{B\psi}|\widetilde{\phi}^{\perp}_{B\psi}}}}}=0.\label{A3}
 \end{align}
 Now using $\braket{\widetilde{\phi}^{\perp}_{X\psi}|\widetilde{\phi}^{\perp}_{X\psi}}= \left(\hspace{-0mm}\braket{\Delta{X}}^{\phi}_{\psi}\hspace{-0mm}\right)^2$, where X=\{A,B\}; the Eq. (\ref{A3}) becomes
 \begin{align}
 {{\braket{\widetilde{\phi}^{\perp}_{B\psi}|\widetilde{\phi}^{\perp}_{A\psi}}}}=\pm i\braket{\Delta{A}}^{\phi}_{\psi}\braket{\Delta{B}}^{\phi}_{\psi}.\label{A4}
 \end{align}
 Finally, use Eqs. (\ref{A4}) and (\ref{SDPPS-3}) in Eq. (\ref{A1}) to obtain the condition (\ref{URPPS-7}).

\section{}\label{B}
 To show that the uncertainty relation  (\ref{URPPS-1}) or (\ref{URPPS-5}) is also valid for  mixed pre-selected state $\rho$, we consider the  following operator 
 \begin{align}
 T=A_0^{\phi}+(\gamma+i\epsilon)B_0^{\phi},\label{B1}
 \end{align}
 where $A_0^{\phi}=A-A\ket{\phi}\bra{\phi}$, $B_0^{\phi}=B-B\ket{\phi}\bra{\phi}$ and $\gamma$, $\epsilon$ are some real parameters. Now for any operator $T$, the inequality 
 \begin{align}
Tr(\rho TT^{\dagger})\geq 0,\label{B2}
 \end{align}
  holds. Using Eq. (\ref{B1}), we have 
 \begin{align}
Tr(\rho TT^{\dagger})=&\left(\braket{\Delta{A}}^{\phi}_{\rho}\right)^2 + \left(\gamma^2+\epsilon^2\right)\left(\braket{\Delta{B}}^{\phi}_{\rho}\right)^2 \nonumber\nonumber\\
&+\gamma\left(\braket{\{A,B\}}_{\rho}-2\mathcal{R}e W_{AB}\right)\nonumber\\
&-i\epsilon\left(\braket{[A,B]}_{\rho}-2\mathcal{I}m W_{AB}\right)  \geq 0,\label{B3}
 \end{align}
 where $\left(\braket{\Delta{A}}^{\phi}_{\rho}\right)^2=Tr(\rho A_0^{\phi} {A_0^{\phi}}^{\dagger})$ is defined in Eq. (\ref{MixedPPS-1}), $\braket{[A,B]}_{\rho}=Tr(\rho[A,B])$, $\braket{\{A,B\}}_{\rho}=Tr(\rho\{A,B\})$  and $W_{AB}=Tr(\Pi_{\phi}B\rho A)$. Now one finds the quantity $Tr(\rho TT^{\dagger})$ is minimum for $\gamma=-\frac{\braket{\{A,B\}}_{\rho}-2\mathcal{R}e W_{AB}}{2\left(\braket{\Delta{B}}^{\phi}_{\rho}\right)^2}$ and  $\epsilon=\frac{i(\braket{[A,B]}_{\rho}-2i\mathcal{I}m W_{AB})}{2\left(\braket{\Delta{B}}^{\phi}_{\rho}\right)^2}$. Hence, $min_{\gamma,\epsilon}$$Tr(\rho TT^{\dagger})\geq 0$ becomes
 \begin{align}
 \left(\braket{\Delta{A}}^{\phi}_{\rho}\right)^2 \left(\braket{\Delta{B}}^{\phi}_{\rho}\right)^2 &\geq  \left[\frac{1}{2i}\braket{[A,B]}_{\rho}-\mathcal{I}m W_{AB}\right]^2 \nonumber\\
 &\hspace{1.5mm}+ \left[\frac{1}{2}\braket{\{A,B\}}_{\rho}-\mathcal{R}e W_{AB}\right]^2.\label{B4}
\end{align}
By discarding  the second term which is a positive number in the right hand side of Eq. (\ref{B4}), the uncertainty relation (\ref{MixedPPS-5}) is achieved.

\section{}\label{C}
Here we show the proofs of all the \emph{Lemmas}  to detect mixedness of an unknown state in qubit, qutrit, qubit-qubit and qubit-qutrit systems. Let us recall the mathematical expression of the statement  of the \emph{Proposition 3} which is given by 
\begin{align}
\braket{\Delta A_w}^{\phi}_{\rho}\geq \braket{\Delta A}^{\phi}_{\rho}.\label{C1}
\end{align} \par
In the following, we will use Eq. (\ref{C1}) to prove all the \emph{Lemmas}. \par

The general form of a mixed state is $\rho=\sum_{i}p_i\ket{\psi_i}\bra{\psi_i}$ and the condition for which the equality of Eq. (\ref{C1}) holds is (see  proof of the \emph{Proposition 3}) 
\begin{align}
\braket{\phi|A|\psi_i}=\lambda\braket{\phi|\psi_i},\label{C2}
\end{align}
  where `$\lambda$' is some constant which  depends on  the index of $\ket{\phi}$ (e.g., for $\ket{\phi_k}$, it is   $\lambda_k$).\\
\textbf{\emph{ The proof of \emph{Lemma 1:}}}

\begin{proof}
We first assume that each $\ket{\psi_i}$ is distinct and hence from Eq. (\ref{C2}), we have a set of equations 
\begin{align}
\bra{\phi}(A-\lambda I)\ket{\psi_i}=0,\label{C3}
\end{align}
for each $\ket{\psi_i}$. Denote the unnormalized state vector $\ket{\widetilde{\phi}^{\lambda}_{A}}=(A-\lambda I)\ket{\phi}$.  As  $\ket{\widetilde{\phi}^{\lambda}_{A}}$ is a  unnormalized state vector different from $\ket{\phi}$ and  the $\ket{\psi_i}$, $\forall i$ are orthogonal to $\ket{\widetilde{\phi}^{\lambda}_{A}}$, it implies that \{$\ket{\psi_i}\}_{i=1}$ are  confined in one  dimensional  Hilbert space.  Hence  each   $\ket{\psi_i}$ is the  same initially prepared state that is $\rho$ is  a pure state in a qubit system.
\end{proof}
\hspace{-3.6mm}\textbf{\emph{ The proof of \emph{Lemma 2:}}}
\begin{proof}
The qubit argument can not be generalized for the higher dimensional systems. The reason is simply because in three dimensional Hilbert space (for example) all the $\ket{\psi_i}$ can be confined in a {two dimensional subspace of the Hilbert space} which is orthogonal to $\ket{\widetilde{\phi}^{\lambda}_{A}}$. 
To make an ``if and only if" condition, we consider the orthogonal basis $\{\ket{\phi_k}\}_{k=1}^3$ as valid post-selections. Here valid post-selections are those post-selections for which  weak values are defined.\par
As there are three post-selections in three dimensional Hilbert space, we have three sets of equations like (\ref{C2}) for the equality of the inequality (\ref{C1})
\begin{align}
\{\bra{\phi_1}(A-\lambda_1 I)\ket{\psi_i}=0\}_{i=1},\label{C4}\\
\{\bra{\phi_2}(A-\lambda_2 I)\ket{\psi_i}=0\}_{i=1},\label{C5}\\
\{\bra{\phi_3}(A-\lambda_3 I)\ket{\psi_i}=0\}_{i=1}.\label{C6}
\end{align}
Now, there are three possibilities which is implied by (\ref{C4}), (\ref{C5}) and (\ref{C6}):\\
($i$)  The  state vectors   \{$\ket{\widetilde{\phi}^{\lambda_k}_{kA}}=(A-\lambda_k I)\ket{\phi_k}\}^{3}_{k=1}$ span the whole 3-dimensional Hilbert space $\mathcal{H}$,\\
($ii$)  \{$\ket{\widetilde{\phi}^{\lambda_k}_{kA}}\}^{3}_{k=1}$ span a 2-dimensional Hilbert space $\mathcal{H}$,\\
($iii$)  \{$\ket{\widetilde{\phi}^{\lambda_k}_{kA}}\}^{3}_{k=1}$ span a 1-dimensional Hilbert space $\mathcal{H}$.\par
Below, we will show that possibility-($i$) is discarded naturally whereas to discard possibility-($iii$), we need a condition on observable $A$ and post-selection $\ket{\phi}$. Then, possibility-($ii$) will  automatically indicate   that all the \{$\ket{{\psi_i}}\}_{i=1}$ are same i.e., $\rho$ is pure.\par
To start with possibility-($i$), let's assume that possibility-($i$) is true, then \{$\ket{\psi_i}\}_{i=1}$ has to be orthogonal to  \{$\ket{\widetilde{\phi}^{\lambda_k}_{kA}}\}^{3}_{k=1}$  according to     (\ref{C4}), (\ref{C5}) and  (\ref{C6}) implying $\ket{\psi_i}=0$ $\forall$ $i$ i.e., $\rho=0$.  So we discard this possibility.\par
Possibility-($iii$) implies
\begin{align}
\mathcal{N}_1(A-\lambda_1I)\hspace{-1mm}\ket{\phi_1}=\mathcal{N}_2(A-\lambda_2 I)\hspace{-1mm}\ket{\phi_2}=\mathcal{N}_3(A-\lambda_3 I)\hspace{-1mm}\ket{\phi_3},\label{C7}
\end{align}
along z-axis (for example)
and hence  \{$\ket{{\psi_i}}\}_{i=1}$ span 2-dimensional xy-plane (see Fig. 4). Here $\mathcal{N}_k$ are nomalization constants.
Now  the inner product of (\ref{C7}) with $\ket{\phi_1}$, $\ket{\phi_2}$ and $\ket{\phi_3}$, respectively  gives
\begin{align}
\mathcal{N}_1(\braket{\phi_1|A|\phi_1}-\lambda_1)=\mathcal{N}_2\braket{\phi_1|A|\phi_2}=\mathcal{N}_3\braket{\phi_1|A|\phi_3},\label{C8}\\
\mathcal{N}_1\braket{\phi_2|A|\phi_1}=\mathcal{N}_2(\braket{\phi_2|A|\phi_2}-\lambda_2)=\mathcal{N}_3\braket{\phi_2|A|\phi_3},\label{C9}\\
\mathcal{N}_1\braket{\phi_3|A|\phi_1}=\mathcal{N}_2\braket{\phi_3|A|\phi_2}=\mathcal{N}_3(\braket{\phi_3|A|\phi_3}-\lambda_3).\label{C10}
\end{align}
Now, the  particular choice
\begin{align}
\braket{\phi_1|A|\phi_2}=0\label{C11}
\end{align}
implies that Eq. (\ref{C8}) and  (\ref{C9})   do not hold if $\braket{\phi_1|A|\phi_3}\neq 0$ and $\braket{\phi_2|A|\phi_3}\neq 0$, respectively. If either of  Eq. (\ref{C8}) and  (\ref{C9})   does not hold then possibility-($iii$) is discarded. But, if  $\braket{\phi_1|A|\phi_3}= 0$ and $\braket{\phi_2|A|\phi_3}= 0$, then we have to proceed further. Note that,  by setting $\braket{\phi_1|A|\phi_2}=0$ from Eq. (\ref{C11}),  $\braket{\phi_1|A|\phi_3}= 0$ and $\braket{\phi_2|A|\phi_3}= 0$ in   Eqs. (\ref{C8}), (\ref{C9}) and  (\ref{C10}), we have 
\begin{align}
\lambda_k=\braket{\phi_k|A|\phi_k} \hspace{5mm}for \hspace{5mm} k=1,2,3.\label{C12}
\end{align}
Now, it is easy to see that with the values of $\lambda_k$  from Eq. (\ref{C12}),  \{$\braket{\phi_k|\widetilde{\phi}^{\lambda_k}_{kA}}=0\}^{3}_{k=1}$ holds. This implies that \{$\ket{\widetilde{\phi}^{\lambda_k}_{kA}}\}^{3}_{k=1}$ can not be  confined in one dimensional Hilbert space i.e., along a particular axis and in our assumption it is   the z-axis. But according to Eq. (\ref{C7}), \{$\ket{\widetilde{\phi}^{\lambda_k}_{kA}}\}^{3}_{k=1}$ are along the z-axis. Hence it shows the contradiction and we discard the possibility-($iii$) when   the condition   $\braket{\phi_1|A|\phi_2}=0$ is considered.\par
Finally the  possibility-($ii$) implies that \{$\ket{\psi_i}\}_{i=1}$ must be  spanned in 1-dimensional Hilbert space $\mathcal{H}$ that is,   each   $\ket{\psi_i}$ is the  same initially prepared state which is a pure state. \par
So, we conclude that  if for an observable $A$ and a complete orthonormal basis $\{\ket{\phi_k}\}_{k=1}^3$ (to be used as post-selected states) of any three-level quantum system (i.e., a qutrit), the  condition  $\braket{\phi_1|A|\phi_2} = 0$  is considered,  then equality in Eq. (\ref{C1}) holds good if and only if the pre-selected state $\rho$ is pure.\par
\end{proof}
\hspace{-3.6mm} \textbf{\emph{The proof of \emph{Lemma 3:}}}
\begin{proof}
 For this bipartite system, we consider the observable and the post-selection to be  $A\otimes I$  and  $\ket{\phi_{AB}}=\ket{\phi_A}\ket{\phi_B}$, respectively. The standard deviations defined in  Eq. (\ref{MixedPPS-1}) and (\ref{MixedPPS-3}) for the given bipartite state $\rho$ become 
 \begin{align}
 \left(\hspace{-.8mm}\braket{\Delta {(\hspace{-.5mm}A\hspace{-.5mm}\otimes \hspace{-.5mm}I)_w}}^{\phi\hspace{-.5mm}_{A\hspace{-.5mm}B}}_{\rho}\hspace{-1mm}\right)^2\hspace{-1.5mm}&=Tr[\hspace{-.5mm}(\hspace{-.5mm}A\hspace{-.5mm}\otimes \hspace{-.5mm} I)^2\hspace{-.5mm}\rho]- \frac{|\hspace{-1mm}\braket{\phi\hspace{-.5mm}_{A\hspace{-.5mm}B}|\hspace{-.5mm}(\hspace{-.5mm}A\hspace{-.5mm}\otimes\hspace{-.5mm} I)\rho|\phi\hspace{-.5mm}_{A\hspace{-.5mm}B}}\hspace{-1mm}|^2}{\braket{\phi\hspace{-.5mm}_{A\hspace{-.5mm}B}|\rho|\phi\hspace{-.5mm}_{A\hspace{-.5mm}B}}}\nonumber\\
  &=Tr[A^2\rho_A] - \frac{|\hspace{-1mm}\braket{\phi_A|A\rho_A^{\phi_B}|\phi_A}\hspace{-1mm}|^2}{\braket{\phi_A|\rho_A^{\phi_B}|\phi_A}},\label{C13}
 \end{align}
 
 \begin{align}
 \hspace{-1mm}\left(\hspace{-.8mm}\braket{\Delta {(\hspace{-.5mm}A\hspace{-.5mm}\otimes \hspace{-.5mm}I)}}^{\phi\hspace{-.5mm}_{A\hspace{-.5mm}B}}_{\rho}\hspace{-1mm}\right)^2\hspace{-1.5mm}&=Tr[\hspace{-.5mm}(\hspace{-.5mm}A\hspace{-.5mm}\otimes \hspace{-.5mm} I)^2\hspace{-.5mm}\rho]\hspace{-.5mm}- \hspace{-.5mm}\braket{\phi\hspace{-.5mm}_{A\hspace{-.5mm}B}|\hspace{-.5mm}(\hspace{-.5mm}A\hspace{-.5mm}\otimes\hspace{-.5mm} I)\rho(\hspace{-.5mm}A\hspace{-.5mm}\otimes\hspace{-.5mm}I)\hspace{-.5mm}|\phi\hspace{-.5mm}_{A\hspace{-.5mm}B}}\nonumber\\
 &=Tr[A^2\rho_A] - \braket{\phi_A|A\rho_A^{\phi_B}A|\phi_A},\label{C14}
 \end{align}
respectively, where $\rho_A^{\phi_B}=\braket{\phi_B|\rho|\phi_B}$ is the collapsed density operator of the subsystem A  when a projection operator $\Pi_{\phi_B}=\ket{\phi_B}\bra{\phi_B}$ is measured in the subsystem B. In a qubit-qubit system, the subsystem  A is two dimensional and  hence $\braket{\Delta {(\hspace{-.5mm}A\hspace{-.5mm}\otimes \hspace{-.5mm}I)_w}}^{\phi\hspace{-.5mm}_{A\hspace{-.5mm}B}}_{\rho}$ from Eq. (\ref{C13}) and $\braket{\Delta {(\hspace{-.5mm}A\hspace{-.5mm}\otimes \hspace{-.5mm}I)}}^{\phi\hspace{-.5mm}_{A\hspace{-.5mm}B}}_{\rho}$ from Eq. (\ref{C14}) are equal ``if and only if" $\rho_A^{\phi_B}$ is pure. Now, $\rho_A^{\phi_B}$ being pure can be from $\rho$ being both pure and mixed.  If $\rho$ is pure then  $\rho_A^{\phi_B}$ is always pure but if  $\rho$ mixed, then it is easy to see that $\rho_A^{\phi_B}$ is pure only when  $\rho=\sum_{i=1}^2 p_i \ket{\psi_A^i}\bra{\psi_A^i}\otimes \ket{\phi_B^i}\bra{\phi_B^i}$, where $\ket{\phi_B^1}=\ket{\phi_B}$ and $\sum^2_{i=1} \ket{\phi_B^i}\bra{\phi_B^i}=I$. So, let us consider another  post-selection $\ket{\phi^{\prime}_B}$ (which  is not orthogonal to  \{$\ket{\phi^{i}_B}\}^{2}_{i=1}$)  and if we find $\rho_A^{\phi^{\prime}_B}$ to be pure which is equivalent to the equality of  $\braket{\Delta {(\hspace{-.5mm}A\hspace{-.5mm}\otimes \hspace{-.5mm}I)_w}}^{\phi^{\prime}\hspace{-.5mm}_{A\hspace{-.5mm}B}}_{\rho}$  and $\braket{\Delta {(\hspace{-.5mm}A\hspace{-.5mm}\otimes \hspace{-.5mm}I)}}^{\phi^{\prime}\hspace{-.5mm}_{A\hspace{-.5mm}B}}_{\rho}$, then we are sure that the bipartite state $\rho$ is a pure state (due to the virtue of qubit system discussed above), where $\ket{\phi^{\prime}\hspace{-.5mm}_{A\hspace{-.5mm}B}}=\ket{\phi_A\phi_B^{\prime}}$.\par
{So, here is the conclusion}: Consider any  two non orthogonal post-selections  $\ket{\phi_B}$ and $\ket{\phi^{\prime}_B}$ in the subsystem B.  For any observable $A$,  equality of $\braket{\Delta {(\hspace{-.5mm}A\hspace{-.5mm}\otimes \hspace{-.5mm}I)_w}}^{\phi\hspace{-.5mm}_{A\hspace{-.5mm}B}}_{\rho}$ and $\braket{\Delta {(\hspace{-.5mm}A\hspace{-.5mm}\otimes \hspace{-.5mm}I)}}^{\phi\hspace{-.5mm}_{A\hspace{-.5mm}B}}_{\rho}$ and separately of  $\braket{\Delta {(\hspace{-.5mm}A\hspace{-.5mm}\otimes \hspace{-.5mm}I)_w}}^{\phi^{\prime}\hspace{-.5mm}_{A\hspace{-.5mm}B}}_{\rho}$ and $\braket{\Delta {(\hspace{-.5mm}A\hspace{-.5mm}\otimes \hspace{-.5mm}I)}}^{\phi^{\prime}\hspace{-.5mm}_{A\hspace{-.5mm}B}}_{\rho}$  hold \emph{only} when the $2\otimes 2$ pre-selected state  $\rho$ is pure.
\end{proof}
\hspace{-3.6mm} \textbf{\emph{The proof of \emph{Lemma 4:}}}
\begin{proof}
 The treatment above with the condition of the  qutrit system, we have the conclusion:
 if for an observable $A$ and any complete orthonormal basis $\{\ket{\phi^k_A}\}_{k=1}^3$ (to be used as post-selected states) for a qutrit,  the condition $\braket{\phi_A^3|A|\phi_A^1} = 0$ is considered,
then equality of $\braket{\Delta {(\hspace{-.5mm}A\hspace{-.5mm}\otimes \hspace{-.5mm}I)_w}}^{\phi\hspace{-.5mm}_{A\hspace{-.5mm}B}}_{\rho}$ and $\braket{\Delta {(\hspace{-.5mm}A\hspace{-.5mm}\otimes \hspace{-.5mm}I)}}^{\phi\hspace{-.5mm}_{A\hspace{-.5mm}B}}_{\rho}$ and separately of $\braket{\Delta {(\hspace{-.5mm}A\hspace{-.5mm}\otimes \hspace{-.5mm}I)_w}}^{\phi^{\prime}\hspace{-.5mm}_{A\hspace{-.5mm}B}}_{\rho}$ and $\braket{\Delta {(\hspace{-.5mm}A\hspace{-.5mm}\otimes \hspace{-.5mm}I)}}^{\phi^{\prime}\hspace{-.5mm}_{A\hspace{-.5mm}B}}_{\rho}$  hold  if and only if the $3 \otimes 2$ pre-selected  state $\rho$ is pure, where $\ket{\phi^{\prime}\hspace{-.5mm}_{A\hspace{-.5mm}B}}=\ket{\phi_A\phi_B^{\prime}}$.
\end{proof}


\bibliographystyle{apsrev4-1}


\begin{thebibliography}{}
\bibitem{Explanation-2}
   {Let the pre-selection be a mixed state $\rho=\sum_{i}p_i\ket{\psi_i}\bra{\psi_i}$. Then the generalization of Eq. (\ref{SDPPS-1}) for the case of $\rho$ is given by $A\rho A=\sum_i p_iA\ket{\psi_i}\bra{\psi_i}A=\sum_ip_i[|\braket{\phi|A|\psi_i}|^2\ket{\phi}\bra{\phi}+(\braket{\Delta A}_{\psi_i}^{\phi})^2\ket{\phi^{\perp}_{A\psi_i}}\bra{\phi^{\perp}_{A\psi_i}}]+\sum_ip_i[\braket{\phi|A|\psi_i}\braket{\Delta A}_{\psi_i}^{\phi}\ket{\phi}\bra{\phi^{\perp}_{A\psi_i}}+c.c]$, where we have used Eq. (\ref{SDPPS-1}) for each $\ket{\psi_i}$.  Now, to see by how much amount the output state deviates from the projection $\Pi_{\phi}=\ket{\phi}\bra{\phi}$  due to the action of $A$ on $\rho$ - similar to the case of $\rho$ being a pure state $\ket{\psi}\bra{\psi}$, as evident from Eq. (\ref{SDPPS-1}), we have to calculate $Tr[A\rho A(I-\Pi_{\phi})]$ and this is given by $(\braket{\Delta A}^{\phi}_{\rho})^2$, where  $(\braket{\Delta A}^{\phi}_{\rho})^2=\sum_i p_i(\braket{\Delta A}_{\psi_i}^{\phi})^2=\sum_ip_i[\braket{\psi_i|A^2|\psi_i}-|\braket{\phi|A|\psi_i}|^2]=Tr(A^2\rho)-\braket{\phi|A\rho A|\phi}$. Hence, it is legitimate to call $\braket{\Delta A}^{\phi}_{\rho}$  as a standard deviation of $A$ in the PPS system.}
\bibitem{example-1}
  {Let $\ket{\psi}$ be the eigenstate of the observable $A$, then the first inequality  in MPUR \cite{Maccone-Pati} becomes $\braket{\Delta{A}}^{2}_{\psi} + \braket{\Delta{B}}^{2}_{\psi} \geq \pm i\braket{\psi|[A,B]|\psi} + \left|\braket{\psi|(A\pm i B)|\psi^{\perp}}\right|^2=
\left|\braket{\psi|B|\psi^{\perp}}\right|^2=\braket{\Delta{B}}^{2}_{\psi}$. Here we have used the fact that in 2D, arbitrary  states $\ket{\psi^{\perp}}$ orthogonal to $\ket{\psi}$ and $\ket{\psi^{\perp}_{B}}$ are same upto a phase factor and hence  $\left|\braket{\psi|B|\psi^{\perp}}\right|=\left|\braket{\psi^{\perp}_{B}|B|\psi}\right|=\braket{\Delta{B}}_{\psi}$ (see Eq. (\ref{RHUR-2})). So, the inequality becomes $\braket{\Delta{B}}^{2}_{\psi}\geq \braket{\Delta{B}}^{2}_{\psi}$}.
\bibitem{example-2}
  {   The second inequality in MPUR \cite{Maccone-Pati} is given by  $\braket{\Delta{A}}^{2}_{\psi} + \braket{\Delta{B}}^{2}_{\psi} \geq \frac{1}{2}\left|\braket{\psi^{\perp}_{A+B}|(A+B)|\psi}\right|^2$, where $\ket{\psi^{\perp}_{A+B}}=(1/\braket{\Delta{(A+B)}}_{\psi})(A+B-\braket{A+B}_{\psi})\ket{\psi}$ and $\braket{\Delta{(A+B)}}_{\psi}^2=\braket{(A+B)^2}_{\psi}-\braket{A+B}^2_{\psi}$. If $\ket{\psi}$ is the eigenstate of the observable $A$, then  $\braket{\Delta{(A+B)}}_{\psi}=\braket{\Delta{B}}_{\psi}$ and using the expression of $\ket{\psi^{\perp}_{A+B}}$ in the above inequality, we have $\braket{\Delta{B}}_{\psi}^{2}\geq \frac{1}{2}\braket{\Delta{B}}_{\psi}^{2}$}.
\end{thebibliography}

\end{document}